\documentclass[11pt,letterpaper]{article}

\usepackage[paper=letterpaper, margin=1in]{geometry}
\usepackage{lmodern}
\usepackage[T1]{fontenc}
\usepackage[utf8]{inputenc}
\usepackage[tbtags]{amsmath}
\allowdisplaybreaks
\usepackage{amssymb,amsthm,mathtools, bbm}
\usepackage{thmtools, thm-restate}
\usepackage{physics}
\usepackage{hyperref}
\usepackage[svgnames]{xcolor}
\definecolor{links}{RGB}{11, 85, 255}
\definecolor{cites}{RGB}{0, 200, 0}
\definecolor{urls}{RGB}{255, 116, 0}
\hypersetup{colorlinks={true},linkcolor={links},citecolor=[named]{cites},urlcolor={urls}}
\usepackage[numbers,compress]{natbib}
\usepackage{doi}
\usepackage{tikz} 
\usepackage{pgfplots}
\pgfplotsset{compat=1.14}
\usepgfplotslibrary{fillbetween}
\usepackage{hyphenat}
\usepackage[font=footnotesize]{caption}
\usepackage{subcaption}
\usepackage{appendix}
\usepackage{nicefrac}
\usepackage{comment}

\usepackage{algorithm}
\usepackage[noend]{algorithmic}

\usepackage[bibliography=common, appendix=inline]{apxproof}

\newcommand{\cS}{\mathcal{S}}

\newcommand{\R}{\mathbb{R}}

\newcommand{\OPT}{\textup{OPT}}
\newcommand{\ALG}{\textup{ALG}}

\newcommand{\ind}[1]{\mathbbm{1}\left\{#1\right\}}

\newcommand{\Ccal}{\mathcal{C}}

\newcommand{\clfrac}[2]{\left\lceil\frac{#1}{#2}\right\rceil}

\newcommand{\act}[1]{#1_{act}}
\newcommand{\excendp}[1]{#1_{o}}  \newcommand{\indicator}[1]{\mathbbm{1}[#1]}

\newtheorem{theorem}{Theorem}[section]
\newtheorem*{theorem*}{Theorem}
\newtheorem{lemma}[theorem]{Lemma}
\newtheorem{proposition}[theorem]{Proposition}

\newtheorem{definition}[theorem]{Definition}

\theoremstyle{plain}
\newtheoremrep{claimapp}[theorem]{Claim}
\theoremstyle{plain}
\newtheoremrep{theoremapp}[theorem]{Theorem}
\theoremstyle{plain}
\newtheoremrep{lemmaapp}[theorem]{Lemma}
\theoremstyle{plain}
\newtheoremrep{propositionapp}[theorem]{Proposition}
\theoremstyle{plain}
\newtheoremrep{corollaryapp}[theorem]{Corollary} 
\usepackage{ifthen}
\usepackage{tikz}
\usetikzlibrary{calc}
\usetikzlibrary{intersections}
\usetikzlibrary{shapes.geometric}

\definecolor{zielony}{rgb}{0.0, 0.65, 0.31}
\colorlet{srednizielony}{zielony!25}
\colorlet{jasnypomarancz}{orange!20}
\definecolor{babyblue}{rgb}{0.54, 0.81, 0.94}

\colorlet{niebieski_mazak}{babyblue!50}
\colorlet{zielony_mazak}{srednizielony!70}
\colorlet{pomaranczowy_mazak}{jasnypomarancz!70} \def\widt {2.5}
\def\heig {3.5}

\newenvironment{graph}[0]{
    \begin{center}
    \begin{scriptsize}
    \begin{tikzpicture}[
scale=0.55,   every node/.style={scale=0.8},
        every text node part/.style={align=center},
graph-node/.style={circle, draw=black, line width=0.5pt},
    ]
}{
    \end{tikzpicture}
    \end{scriptsize}
    \end{center}
}

\newcommand{\drawnode}[5]{
    \node[graph-node, draw=#4, fill=#5] (#1) at (#2) {$#3$};
}

\newcommand{\drawnodesubs}[6]{
    \node[graph-node, draw=#5, fill=#6] (#1) at (#2) {$#3_#4$};
}

\newcommand{\drawtriangle}[4]{
    \path [draw=#3,fill=#4] (#1) -- ($(#1)-(#2*\widt,#2*\heig)$) -- ($(#1)-(0,#2*\heig)+(#2*\widt,0)$) -- cycle;
}

\newcommand{\drawsnakeline}[4]{
  \path [draw=#4, line width=3pt] (#1) -- (#2);
  \path [draw=#3, dashed] (#1) -- (#2);
} 
\begin{document}

\begin{titlepage}
\def\thepage{}
\thispagestyle{empty}

\title{Universal Optimization for Non-Clairvoyant Subadditive Joint Replenishment}

\author{Tomer Ezra\thanks{Harvard University,
Email:  \texttt{tomer@cmsa.fas.harvard.edu
}} \and Stefano Leonardi\thanks{Sapienza University of Rome, Email:  \texttt{leonardi@diag.uniroma1.it}} \and Michał Pawłowski\thanks{MIMUW, University of Warsaw and IDEAS NCBR, Email: \texttt{michal.pawlowski196@gmail.com}} \and Matteo Russo\thanks{Sapienza University of Rome, Email:  \texttt{mrusso@diag.uniroma1.it}} \and Seeun William Umboh\thanks{School of Computing and Information Systems, 
The University of Melbourne and ARC Training Centre in Optimisation Technologies, Integrated Methodologies, and Applications (OPTIMA), Email: \texttt{william.umboh@unimelb.edu.au}}}

\date{}

\maketitle
\begin{abstract}
The online joint replenishment problem (JRP) is a fundamental problem in the area of online problems with delay. Over the last decade, several works have studied generalizations of JRP with different cost functions for servicing requests. Most prior works on JRP and its generalizations have focused on the clairvoyant setting. 
Recently, Touitou \cite{Touitou23} developed a non-clairvoyant framework that provided an $O(\sqrt{n \log n})$ upper bound for a wide class of generalized JRP, where $n$ is the number of request types.

We advance the study of non-clairvoyant algorithms by providing a simpler, modular framework that matches the competitive ratio established by Touitou for the same class of generalized JRP. 
Our key insight is to leverage universal algorithms for Set Cover to approximate arbitrary monotone subadditive functions using a simple class of functions termed \textit{disjoint}. This allows us to reduce the problem to several independent instances of the TCP Acknowledgement problem, for which a simple 2-competitive non-clairvoyant algorithm is known.  The modularity of our framework is a major advantage as it allows us to tailor the reduction to specific problems and obtain better competitive ratios. In particular, we obtain tight $O(\sqrt{n})$-competitive algorithms for two significant problems: Multi-Level Aggregation and Weighted Symmetric Subadditive Joint Replenishment. We also show that, in contrast, Touitou's algorithm is $\Omega(\sqrt{n \log n})$-competitive for both of these problems.

\end{abstract}

\end{titlepage}

\section{Introduction}
Online problems with delay have received much attention in the last few years. An important family of online problems with delay consists of the Joint Replenishment Problem (JRP) and its variants. A typical instance consists of a sequence of requests that arrive over time. Each request can be one of $n$ request types, and the cost of serving a set of requests is a subadditive\footnote{A set function over a ground set $U$ is \emph{subadditive} if $f(A) + f(B) \geq f(A \cup B)$ for every $A,B \subseteq U$.} function of their types. We assume that the algorithm has oracle access to the service cost function. Requests do not need to be served on arrival but each request accumulates a delay cost while unserved. In particular, each request $q$ has an associated delay cost function $d_q$ and its delay cost is $d_q(t)$ if it is served at time $t$. The goal of the problem is to serve all requests minimizing the total service cost and delay cost. An important special case is the deadline case; this is when requests do not incur delay cost but instead must be served by some given time. 
We call this family of problems \emph{Subadditive JRP}.

These problems can be studied under the clairvoyant and non-clairvoyant settings. In the \emph{clairvoyant} setting, when a request $q$ arrives, the algorithm is given the entire delay cost function $d_q$ (or its deadline in the case of deadlines). In contrast, in the \emph{non-clairvoyant} setting, the algorithm only knows of the delay cost accumulated so far. In the case of deadlines, the algorithm only knows whether the request's deadline is now (and must be served immediately) or later. 

Most previous works on Subadditive JRP have focused on the clairvoyant setting. Key problems within the family of Subadditive JRP include (in increasing order of generality): TCP Acknowledgement~\cite{dynamictcp, DoolyGS01, BuchbinderJN07}, Joint Replenishment Problem~\cite{BuchbinderKLMS08, BritoKV12, BienkowskiBCJS13, CJRP}, and Multi-Level Aggregation (MLA) \cite{BuchbinderFNT17, AzarT19, BienkowskiBBCDF21, BienkowskiBBCDF20, McMahan-MLA}. For general subadditive service cost functions, deterministic $O(\log N)$ (where $N$ is the number of requests) and $O(\log n)$ upper bounds are known (\cite{CarrascoPSV18} and \cite{AzarT20}, respectively).

There is much less work in the non-clairvoyant setting. For a small number of problems, such as TCP Acknowledgement and Set Cover with Delay~\cite{AzarCKT20}, clairvoyance is not required in the sense that the same competitive ratio can be attained in both the clairvoyant and non-clairvoyant settings. However, Azar et al.~\cite{OSD}'s lower bound for Online Service with Delay (a different family of online problems with delay) can be translated into an $\Omega(\sqrt{n})$ lower bound against deterministic algorithms for JRP, and thus, MLA and Subadditive JRP. In contrast, clairvoyant Subadditive JRP has a $O(\log n)$ competitive ratio \cite{AzarT20}. Recently, Le et al.~\cite{LeUX23} showed that randomization does not help in breaking the $\Omega(\sqrt{n})$ barrier and also developed algorithms for JRP and MLA with matching and nearly-matching upper bounds. Shortly after, Touitou~\cite{Touitou23} presented a general non-clairvoyant framework for Subadditive JRP with a deterministic $O(\sqrt{n \log n})$ competitive ratio. 

 \subsection{Our Results}
Our main contribution is a simple, modular framework for non-clairvoyant Subadditive JRP that matches the current-best competitive ratio of $O(\sqrt{n \log n})$, and yields tight $O(\sqrt{n})$ competitive ratios for the key problems of Multi-Level Aggregation and Weighted Symmetric Subadditive Joint Replenishment. We also show that the framework of Touitou~\cite{Touitou23} is $\Omega(\sqrt{n \log n})$ for these problems. We now formally define these problems and state our results.

\subsubsection{General Framework for Subadditive JRP}

\paragraph*{Subadditive JRP.} We have a set $U$ of $n$ request types and a monotone non-decreasing, subadditive \emph{service function} $f : 2^U \mapsto \mathbb{R}_{\geq 0}$ that satisfies $f(\emptyset) = 0$. Requests $q$ arrive over time. Each request $q$ has a type $h_q \in U$, an arrival time $a_q$, and a non-decreasing, continuous delay function $d_q$. At any point in time, the algorithm can serve a subset $Q$ of the requests that have arrived and incur a service cost of $f(S_Q)$ where $S_Q = \{h_q : q \in Q\}$ is the set of types of $Q$. Let $C_q$ be the time when request $q$ was served. The \emph{delay cost} of request $q$ is $d_q(C_q)$.\footnote{We assume W.L.O.G. that $d_q(a_q) = 0$, i.e.,~serving a request immediately on arrival incurs no delay cost.} The goal is to serve all requests while minimising the sum of the total service and delay costs.

\paragraph{Approximating set functions.} The core idea underlying our framework is the following simple but powerful observation. Given two set functions $f, g$ over the same ground set $U$ of $n$ elements, we say that $g$ is an \emph{$\alpha$-approximation} of $f$ if $f(S) \leq g(S) \leq \alpha f(S)$ for every $S \subseteq U$. Our observation is that for a given subadditive service function $f$, if we can $\alpha$-approximate $f$ by a simpler service function $g$, then we can reduce any instance of Subadditive JRP with service function $f$ to one with $g$ instead. In fact, this leads us to the following simplification of the problem.

\paragraph{Disjoint TCP Acknowledgement.} In Disjoint TCP Acknowledgement, we have a set $U$ of $n$ request types. We also have a partition of $U$ into subsets $S_1, \ldots, S_k$ with costs $c_1, \ldots, c_k$. For a subset $S \subseteq U$, we have $f(S) = \sum_{i = 1}^k c_i \cdot \ind{S_i \cap S \neq \emptyset}$. In other words, we pay $c_i$ for every part $S_i$ that intersects with $S$. Such a function is called a \emph{disjoint service function}. Observe that when $k = 1$, this is equivalent to the TCP Acknowledgement problem; when $k > 1$, this corresponds to several independent instances of TCP Acknowledgement. The 2-competitive algorithm for TCP Acknowledgement of \cite{DoolyGS01} can be easily extended to a 2-competitive algorithm for Disjoint TCP Acknowledgement (see Section~\ref{sec:reduction}).

We now state our main technical lemma.
\begin{restatable}[Reduction Lemma]{lemma}{reduction}
    \label{lem:reduction-to-disjoint}
If there exists a \textbf{disjoint} service function $g$ that $\alpha$-approximates $f$, then there exists a non-clairvoyant algorithm that is $2\alpha$-competitive non-clairvoyant algorithm for every Subadditive JRP instance with service cost function $f$. 
\end{restatable}

A major advantage of our Reduction Lemma is that it reduces the task of designing and analyzing an online algorithm for a Subadditive JRP problem to the much cleaner task of showing that the corresponding service function $f$ can be approximated by a disjoint service function well. In particular, this boils down to finding a partition of the set of request types $U$ into subsets $S_1, \ldots, S_k$, for some $k$, such that the following quantity is small 
\[\max_{S \subseteq U} \frac{\sum_{i = 1}^k f(S_i) \cdot \ind{S_i \cap S \neq \emptyset}}{f(S)}.\]

For general Subadditive JRP, our key insight is that the problem of approximating an arbitrary service function $f$ by a disjoint service function can be reformulated as the Universal Set Cover problem.

\paragraph{Universal Set Cover (USC).} An instance of the Universal Set Cover (USC) problem consists of a universe $U$ of $n$ elements, a collection $\Ccal$ of subsets of $U$, and costs $c(S)$ for each set $S \in \Ccal$. A solution is an assignment $a$ of each element $e$ to a set $a(e) \in \Ccal$. For any subset $X \subseteq U$, define $a(X) = \{a(e) : e \in X\}$. The \emph{stretch} of the assignment $a$ is $\max_{X \subseteq U} c(a(X)) / \OPT(X)$ where $\OPT(X)$ is the cost of the optimal set cover of $X$. 

Jia et al.~\cite{JiaLNRS05} introduced the Universal Set Cover problem and showed that a $O(\sqrt{n \log n})$-stretch assignment can always be found efficiently. We show that this implies that any subadditive service function $f$ can be approximated by a disjoint service function to within a factor of $O(\sqrt{n \log n})$ (Lemma~\ref{lem:usc}). Together with our Reduction Lemma, we get a deterministic $O(\sqrt{n \log n})$-competitive algorithm for Non-Clairvoyant Subadditive JRP, matching the current state-of-the-art~\cite{Touitou23}.

\subsubsection{MLA and Weighted Symmetric Subadditive JRP} 
\label{sec:ourtechniques}

One main technical contribution of the paper is to exploit the inherent structure of the MLA and  Weighted Symmetric Subadditive JRP functions to show that they can be $O(\sqrt{n})$-approximated by disjoint service functions. We then employ the Reduction Lemma to prove tight $O(\sqrt{n})$-competitive ratios for the two corresponding problems.

\paragraph*{Multi-Level Aggregation.} In the Multi-Level Aggregation (MLA) problem, the service function $f$ is defined by a rooted \emph{aggregation tree} $T$, where each node corresponds to a different request type. Let $r$ be the root of $T$ and let $c(v)$ be the cost of node $v$ for each $v \in T$. For a subset $V$ of nodes, $f(V)$ is defined to be the total cost of the nodes in the minimal subtree connecting $V$ to $r$.

\begin{restatable}{theorem}{nonclairvoyantmla}
    \label{thm:non-clairvoyant-mla}
    There exists an efficient deterministic $O(\sqrt{n})$-competitive algorithm for the Non-Clairvoyant Multi-Level Aggregation problem. 
\end{restatable}

To show the above result, given Lemma~\ref{lem:reduction-to-disjoint}, our goal is to find a good partition $P$ of the tree $T$'s nodes into subtrees and subforests (that we refer to as clusters). More precisely, let us use $P$ to define a disjoint service function $g$ where for each subset $V$ of nodes of $T$, $g(V) = \sum_{C \in P : C \cap V \neq \emptyset} f(C)$. 

The crucial idea is to notice that since we aim for the gap of order at most $\sqrt{n}$ between $g$ and $f$, we can see it as $g$ being assigned a budget of roughly $\sqrt{n}f(V)$ to serve $V$ for each subset $V$ of $T$'s nodes. Since the cost that $f$ incured on a set $V$ equals the cost of the minimal subtree connecting all the nodes in $V$ to the root $r$ of $T$, the value of $g(V)$ cannot exceed $\beta \sqrt{n}$ times this cost for some fixed $\beta \in \mathbb{N}$. To achieve this, we generate a partition consisting of two types of clusters. First are the subtrees rooted at ``expensive'' nodes. The intuition is that their cost alone multiplied by $\alpha \sqrt{n}$ for some $\alpha \in \mathbb{N}$ is enough to ``cover'' the cost of both their subtree and the path to $r$. The second type is the clusters that contain more than $\sqrt{n}$ nodes, since there cannot be many of them.

\paragraph*{Weighted Symmetric Subadditive JRP.} In Weighted Symmetric Subadditive JRP, the service function $f$ is a function of the total weight $w(S) = \sum_{i \in S} w_i$ of the set of types being served. In particular, $f$ is a monotone non-decreasing subadditive function with $f(S) = f(w(S))$ and $f(0) = 0$, that satisfies that for every weights $x,y$, $f(x+y)\leq f(x)+f(y)$. We refer to these functions as \textit{weighted symmetric subadditive}. 

\begin{restatable}{theorem}{nonclairvoyantcjrp}
    \label{thm:non-clairvoyant-cjrp}
    There exists an efficient deterministic $O(\sqrt{n})$-competitive algorithm for the Non-Clairvoyant Weighted Symmetric Subadditive Joint Replenishment problem. 
\end{restatable}

As in the MLA case, given Lemma~\ref{lem:reduction-to-disjoint}, our goal is to devise a partitioning algorithm inducing a disjoint service function that $O(\sqrt{n})$-approximates the corresponding weighted symmetric subadditive service cost function. We first consider the special case where each weight equals $1$. In this scenario, the service function $f$ is symmetric and becomes a function of the cardinality of the set of types being served. Consequently, the partition of elements should ideally reflect this symmetry by ensuring equal-sized parts.

Determining the optimal size for each part involves striking a delicate balance. Larger sizes enable us to leverage the subadditivity of $f$ but excessively large sizes incur higher costs for smaller sets. We demonstrate that selecting sets of size $O(\sqrt{n})$ is the optimal tradeoff in worst-case scenarios. Notably, this partition remains effective across all symmetric subadditive functions simultaneously.

Extending this approach to the general case of weighted symmetric subadditive functions involves categorizing elements into weight classes based on powers of 2, ensuring approximate size equivalence, and then partitioning into sets of size $\sqrt{n}$. However, this approach risks generating an excessive number of sets. To address this issue, we devise a partitioning strategy that accommodates light-weight elements first. Then, for heavier-weight elements, we further partition by a factor of 2, provided it is feasible, to achieve a refined division.

\subsubsection{Running time of Algorithms and Reductions}

Regarding the running time of our algorithms, we stress that, in the case of Multi-Level Aggregation and Weighted Symmetric Subadditive JRP, the reductions are executed in polynomial time. However, the reduction for general subadditive functions is executed in exponential time, as we need to create a set for each subset of types.

\subsubsection{Lower bounds on approximating subadditive service functions.} 
Since Non-Clairvoyant MLA and Weighted Symmetric Subadditive JRP have a $\Omega(\sqrt{n})$ lower bound~\cite{OSD,LeUX23}, the Reduction Lemma implies that MLA and Weighted Symmetric Subadditive JRP service functions do not admit $o(\sqrt{n})$-approximation by disjoint service functions. Nevertheless, we also give direct proofs in Sections~\ref{sec:DWMLA} and~\ref{sec:weighted}, respectively.
The latter provides a simpler alternative proof for the $\Omega(\sqrt{n})$ lower bound for unweighted Universal Set Cover shown in \cite{JiaLNRS05}. We also show, in Proposition~\ref{thm:tight}, that Jia et al.'s analysis of their Universal Set Cover algorithm \cite{JiaLNRS05} is tight. Thus, we need a different approach to $o(\sqrt{n \log n})$-approximate arbitrary subadditive service functions by disjoint service functions. Finally, in Proposition~\ref{prop:tight-touitou}, we exhibit an MLA and Weighted Symmetric Subadditive JRP instance where Touitou's algorithm \cite{Touitou23} can only achieve an $\Omega(\sqrt{n \log n})$-approximation to the respective service cost functions.

\subsection{Future Directions} 

Our work leaves several tantalizing open questions. The main open problem is whether subadditive service functions admit better than $O(\sqrt{n \log n})$-approximation by disjoint service functions. This would immediately improve the competitive ratio for general non-clairvoyant Subadditive JRP. It would also be interesting to find better approximations of other interesting subclasses such as XOS and submodular functions.

\begin{toappendix}
    \subsection{Further Related Work}\label{sec:rel-work}

\paragraph{Network Design with Delay.}  Network Design with Delay is very closely related to Subadditive JRP. In Network Design with Delay, we are given a universe of $n$ request types and $m$ items with costs. Each request type $h$ has a corresponding upwards-closed collection $\mathcal{C}_h$ of subsets of items that satisfy it. At any point in time, the algorithm can transmit a set of items. A request of type $h$ is served by a transmission that contains some subset in $\mathcal{C}_h$. Some specific problems are Set Cover with Delay~\cite{CarrascoPSV18,AzarCKT20, Touitou21}, Facility Location~\cite{AzarT19,AzarT20, FL-linear} and other network design problems~\cite{AzarT19,AzarT20}. Network Design with Delay is equivalent to Subadditive JRP as the optimal cost of satisfying a subset of request types is subadditive, and Subadditive JRP can be formulated as Set Cover with Delay with exponentially many sets.

\paragraph{Online problems with delay.}
There has been a lot of work on other online problems with delay as well. In Online Service with Delay, we are given one or multiple servers on a metric space. Requests arrive on points of the metric space and are served when a server is moved to their location. In Online Matching with Delay, we are given an underlying metric space. Requests arrive on points of the metric space and are served when they are matched. Most of the work on Online Service with Delay~\cite{OSD,cachingDelay,hitting-kserver,OSD-line, nonclairvoyantKServer, Touitou23-OSD} and Online Matching with Delay~\cite{emekoriginal,polylog, bipartite, spheres, hemispheres, pd,2sources, concave, convex, DeryckereU23} has been in the clairvoyant setting. Nevertheless, non-clairvoyant algorithms have been designed for Online Service with Delay~\cite{nonclairvoyantKServer} and Online Matching with Delay~\cite{DeryckereU23}.

\paragraph{Approximating subadditive functions.} The approximation of subadditive functions has been a focal point of research, at least since the introduction of the complement-free hierarchy of functions introduced in \cite{LehmannLN06}. This consists of the class of submodular function, which is strictly contained into the XOS class, which in turn is strictly contained in the general subadditive class.\footnote{Several other classes within the submodular class have been considered (e.g. additive, unit-demand, Gross-Substitutes).} As for approximation, it is known that XOS approximates subadditive within a factor of $O(\log(n))$, which is tight \cite{Dobzinski07, BhawalkarR11}. The approximability gap between Submodular and XOS is $\Theta(\sqrt{n})$ \cite{BadanidiyuruDFKNR12, GoemansHIM09}. In a similar vein, \cite{DobzinskiFF21} prove that Gross-Substitute functions (first introduced in \cite{KelsoC82}) cannot approximate submodular set functions within a factor better than $\Omega\left(\frac{\log(n)}{\log\log(n)}\right)$.
In the context of \emph{symmetric} function approximation, \cite{EzraFRS20} show that symmetric subadditive, symmetric XOS and symmetric submodular\footnote{We use the term symmetric submodular to indicate functions that are (monotone) concave in the size of the set.} functions are all $2$-close to each other, which is tight. \end{toappendix} 
\section{Subadditive Joint Replenishment}
In this section, we prove our Reduction Lemma (Lemma~\ref{lem:reduction-to-disjoint}) and apply it to Subadditive JRP.

\subsection{Reduction Lemma}\label{sec:reduction}

We begin by showing that there is a simple deterministic 2-competitive algorithm for Disjoint TCP Acknowledgement via a straightforward extension of the classic algorithm for TCP Acknowledgement of \cite{DoolyGS01}. 

In the following, we use $\lambda$ to denote a service and $Q_\lambda$ to be the set of request types transmitted by $\lambda$. We also use $\OPT$ to mean both the optimal solution and the cost of the optimal solution.

\begin{lemma}\label{lem:disjoint-alg}
    There is a deterministic $2$-competitive algorithm for Disjoint TCP Acknowledgement.
\end{lemma}

\begin{proof}
    Suppose there is a partition of $H$ into subsets $S_1, \ldots, S_k$ with costs $c_1, \ldots, c_k$ and $f(S) = \sum_{i = 1}^k c_i \cdot \ind{S_i \cap S \neq \emptyset}$. Our algorithm works as follows: for each set $S_i$, transmit $S_i$ whenever the pending requests in $S_i$ have accumulated a total delay equal to $c_i$. 

    It is clear that the total service cost of the algorithm is at most its total delay cost. We now show that the latter is at most the cost of the optimal solution. To this end, let us consider the cost of the optimal solution. Suppose that the optimal solution makes a set of services $\Lambda^*$. Let $\Lambda^*_i$ denote the subset of services that transmit a request type in $S_i$. The total service cost of the optimal solution is then 
    \begin{align*}
        \sum_{\lambda \in \Lambda^*} f(Q_\lambda) 
        = \sum_{\lambda \in \Lambda^*} \sum_{i=1}^k c_i \cdot \ind{S_i \cap Q_\lambda \neq \emptyset}
        = \sum_{i=1}^k c_i \cdot |\Lambda^*_i|.
    \end{align*} 

    Define $d^\OPT_q$ and $d^\ALG_q$ to be the delay cost of $q$ in the optimal solution and algorithm's solution, respectively. 
    Let $\OPT_i = c_i \cdot |\Lambda^*_i| + \sum_{q : h_q \in S_i} d^\OPT_q$. This is the total cost that $\OPT$ incurs on requests on $S_i$. Observe that $\OPT = \sum_{i=1}^k \OPT_i$. 
    
    We now show that $\sum_{q : h_q \in S_i} d^\ALG_q \leq \OPT_i$ for each set $S_i$. Suppose that the algorithm transmits $S_i$ at times $t_1, \ldots, t_\ell$. Since every request must be served eventually, no request with type in $S_i$ arrives after $t_\ell$. Consider the intervals $[0, t_1], (t_1, t_2], \ldots (t_{\ell-1}, t_\ell)$. By construction, the delay cost of the algorithm is exactly $\ell c_i$. For each interval $I$, let $Q(I)$ denote the requests with types in $S_i$ that arrived during the interval. During $I$, the optimal solution either transmits a type in $S_i$ or incurs a delay cost of $c_i$ on the requests in $Q(I)$. Since the intervals are disjoint, $\OPT_i \geq \ell c_i$, as desired. 

    The lemma now follows from the fact that the total service cost of the algorithm is exactly its delay cost, which in turn is at most $\OPT$.\end{proof}

We are now ready to prove the Reduction Lemma which we restate here.

\reduction*

\begin{proof}
 Lemma~\ref{lem:disjoint-alg} implies that it suffices to reduce the Subadditive JRP instance to an instance of Disjoint TCP Acknowledgement losing at most a factor of $\alpha$. Let $Q$ be the set of requests of the Subadditive JRP instance and let $\OPT^f$ denote the cost of the optimal solution. Our reduction creates an instance of Disjoint TCP Acknowledgement with the same set of requests but with service cost function $g$. Let $\OPT^g$ denote the cost of the optimal solution to the instance of Disjoint TCP Acknowledgement. We now show that $\OPT^f \leq \OPT^g \leq \alpha \OPT^f$. Let $\Lambda$ be a feasible solution to $Q$, $c^f(\Lambda)$ be its cost in the Subadditive JRP instance and $c^g(\Lambda)$ be its cost in the Disjoint TCP Acknowledgement instance. The delay cost of $\Lambda$ is the same in both instances. The service cost of $\Lambda$ in the Subadditive JRP instance has cost $\sum_{\lambda \in \Lambda} f(Q_\lambda)$ and in the Disjoint TCP Acknowledgement instance, it has cost $\sum_{\lambda \in \Lambda} g(Q_\lambda)$. Since $g$ $\alpha$-approximates $f$, we get that $c^f(\Lambda) \leq c^g(\Lambda) \leq \alpha c^f(\Lambda)$. This implies that $\OPT^f \leq \OPT^g \leq \alpha \OPT^f$, as desired.
\end{proof}

 \subsection{Applying the Reduction Lemma to Subadditive JRP}
We use the Reduction Lemma proved earlier to give a simple deterministic $O(\sqrt{n \log n})$-competitive algorithm for Non-Clairvoyant Subadditive JRP. The main insight is to reduce the problem of showing that an arbitrary service function $f$ can be approximated by a disjoint service function to the Universal Set Cover problem.  

\begin{lemma}
\label{lem:usc}
    Suppose every instance of USC admits a $\alpha$-stretch assignment. Then every subadditive service function $f$ can be $\alpha$-approximated by some disjoint service function $g$.
\end{lemma}

\begin{proof}
    We will construct an instance of USC and use the $\alpha$-stretch assignment to construct $g$. Consider the instance of USC with universe $U = H$, $\Ccal = 2^H$, and $c(S) = f(S)$ for every $S \in \Ccal$. Note that $\OPT(S) = f(S)$ since $f$ is monotone non-decreasing and subadditive.
    
    Let $a$ be an $\alpha$-stretch assignment for this USC instance. Suppose $a(U) = \{S_1, \ldots, S_k\}$. Since $a$ maps each element to a set containing it, we have that $a^{-1}(S_i) \subseteq S_i$. Moreover, $f$ is monontone non-decreasing, so we can assume W.L.O.G. that $a^{-1}(S_i) = S_i$;\footnote{Otherwise, we can assign the elements in the preimage of $S_i$ under $a$, i.e., $a^{-1}(S_i)$, to the preimage itself.} 
thus $S_1, \ldots, S_k$ are disjoint and partition $H$. Define the disjoint service function $g$ with the partition $\{S_1, \ldots, S_k\}$ and costs $c_1, \ldots, c_k$ where $c_i = f(S_i)$. Observe that $g(S) = c(a(S)) \geq \OPT(S) = f(S)$. Since $a$ has $\alpha$-stretch, we get that for every $S$, $f(S) \leq g(S) \leq \alpha f(S)$.
\end{proof}

Jia et al.~\cite{JiaLNRS05} showed that every instance of USC has a $O(\sqrt{n \log n})$-stretch assignment. Together with the above lemma, we get the following theorem.

\begin{theorem}
    For every subadditive service function $f$, there is a disjoint service function $g$ that $O(\sqrt{n \log n})$-approximates $f$.
\end{theorem}

Combining this with the Reduction Lemma yields the desired theorem.

\begin{theorem}
    \label{thm:subadd-general}
    There is a deterministic $O(\sqrt{n \log n})$-competitive algorithm for Non-Clairvoyant Subadditive JRP. 
\end{theorem}

\section{Multi-Level Aggregation}
\label{sec:DWMLA}

In this section, we consider the Multi-Level Aggregation (MLA) problem. 
Let $T = (U, E)$ be a rooted tree defined over the universe $U$ of $n$ request types and let $c: U \mapsto \mathbb{R}_{\geq 0}$ be a cost function assigning weights to the nodes. 
We recall that $c$ determines the service function $f: 2^{U} \mapsto \mathbb{R}_{\geq 0}$ for this problem as $f$ assigns to each subset of nodes $V \subseteq U$ the cost of the minimal subtree that connects all the nodes in $V$ to the root $r$. 
Here, we prove that for every MLA service function $f$, there exists a disjoint service function $g: 2^{U} \mapsto \mathbb{R}_{\geq 0}$ that $O(\sqrt{n})$-approximates $f$. 
In other words, we show that for every MLA instance $(T, c)$, there exists a partition $P_1, \ldots, P_k$ of nodes of $T$ for some $k$ (which defines $g(X)= \sum_{i\in [k]} f(P_i) \cdot \ind{P_i \cap X \neq \emptyset}$ for all $X\subseteq U$), such that for all $V \subseteq U$, it holds that $g(V)/f(V) \leq O(\sqrt{n})$. 
Moreover, one can find such a partition in polynomial time.

\subsection{Notation and Algorithm Overview}

Throughout this section, we assume tree $T$ is the current MLA instance that we work with and thus is known from the context. 
In what follows, we refer to the \emph{maximal subtree} of $T$ rooted at node $v$ and to the \emph{path} connecting $v$ to the root $r$ by simply writing $T(v)$ and $R(v)$, respectively. 
Moreover, to denote these objects with node $v$ excluded, we use $\excendp{T}(v)$ and $\excendp{R}(v)$.
Finally, we let $C(v)$ be the set of $v$'s \emph{children} in $T$.

\begin{figure}[ht]
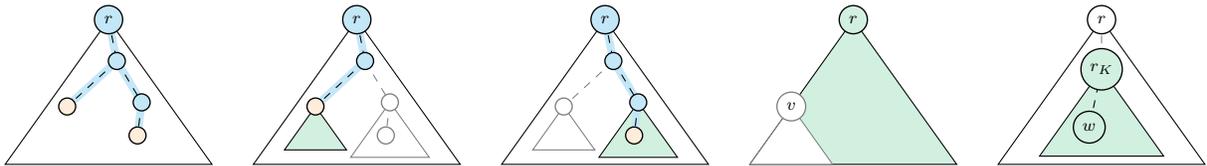

\begin{graph}
    \def\vertdist {6}

    \coordinate (R1) at (0,0);
    \coordinate (V1) at ($(R1)-(-0.2,1)$);
    \coordinate (V2) at ($(V1)-(1.2,1.1)$);
    \coordinate (V3) at ($(V1)-(-0.6,1)$);
    \coordinate (V4) at ($(V3)-(0.1,0.8)$);

    \drawtriangle{R1}{1}{black}{white}

    \drawsnakeline{R1}{V1}{black}{niebieski_mazak}
    \drawsnakeline{V1}{V2}{black}{niebieski_mazak}
    \drawsnakeline{V1}{V3}{black}{niebieski_mazak}
    \drawsnakeline{V3}{V4}{black}{niebieski_mazak}
    
    \drawnode{N1}{R1}{r}{black}{niebieski_mazak}
    \drawnode{N2}{V1}{}{black}{niebieski_mazak}
    \drawnode{N3}{V2}{}{black}{pomaranczowy_mazak}
    \drawnode{N4}{V3}{}{black}{niebieski_mazak}
    \drawnode{N5}{V4}{}{black}{pomaranczowy_mazak}

\coordinate (R2) at (\vertdist,0);
    \coordinate (V1) at ($(R2)-(-0.2,1)$);
    \coordinate (V2) at ($(V1)-(1.2,1.1)$);
    \coordinate (V3) at ($(V1)-(-0.6,1)$);
    \coordinate (V4) at ($(V3)-(0.1,0.8)$);

    \drawtriangle{R2}{1}{black}{white}
    \drawtriangle{V2}{0.3}{black}{zielony_mazak}
    \drawtriangle{V3}{0.38}{gray}{white}

    \drawsnakeline{R2}{V1}{black}{niebieski_mazak}
    \drawsnakeline{V1}{V2}{black}{niebieski_mazak}
    \drawsnakeline{V1}{V3}{gray}{white}
    \drawsnakeline{V3}{V4}{gray}{white}
    
    \drawnode{N1}{R2}{r}{black}{niebieski_mazak}
    \drawnode{N2}{V1}{}{black}{niebieski_mazak}
    \drawnode{N3}{V2}{}{black}{pomaranczowy_mazak}
    \drawnode{N4}{V3}{}{gray}{white}
    \drawnode{N5}{V4}{}{gray}{white}

\coordinate (R3) at (2*\vertdist,0);
    \coordinate (V1) at ($(R3)-(-0.2,1)$);
    \coordinate (V2) at ($(V1)-(1.2,1.1)$);
    \coordinate (V3) at ($(V1)-(-0.6,1)$);
    \coordinate (V4) at ($(V3)-(0.1,0.8)$);

    \drawtriangle{R3}{1}{black}{white}
    \drawtriangle{V2}{0.3}{gray}{white}
    \drawtriangle{V3}{0.38}{black}{zielony_mazak}

    \drawsnakeline{R3}{V1}{black}{niebieski_mazak}
    \drawsnakeline{V1}{V2}{gray}{white}
    \drawsnakeline{V1}{V3}{black}{niebieski_mazak}
    \drawsnakeline{V3}{V4}{black}{niebieski_mazak}
    
    \drawnode{N1}{R3}{r}{black}{niebieski_mazak}
    \drawnode{N2}{V1}{}{black}{niebieski_mazak}
    \drawnode{N3}{V2}{}{gray}{white}
    \drawnode{N4}{V3}{}{black}{niebieski_mazak}
    \drawnode{N5}{V4}{}{black}{pomaranczowy_mazak}
    
\coordinate (R) at (3*\vertdist,0);
    \coordinate (V) at ($(R)+(-1.5,-2.1)$);

    \drawtriangle{R}{1}{black}{zielony_mazak}
    \drawtriangle{V}{0.4}{gray}{white}
    
    \drawnode{N1}{R}{r}{black}{zielony_mazak}
    \drawnode{N2}{V}{v}{gray}{white}

\coordinate (R) at (4*\vertdist,0);
    \coordinate (V) at ($(R)+(0,-1.2)$);
    \coordinate (V2) at ($(V)-(0.3,1.4)$);

    \drawtriangle{R}{1}{black}{white}
    \drawtriangle{V}{0.6}{black}{zielony_mazak}

    \drawsnakeline{R}{V}{gray}{white}
    \path[draw=black, dashed] (V) -- (V2);
    
    \drawnode{N1}{R}{r}{black}{white}
    \drawnodesubs{N2}{V}{r}{K}{black}{zielony_mazak}
    \drawnode{N3}{V2}{w}{black}{zielony_mazak}
\end{graph} \caption{In the first three figures we show the costs of serving the orange nodes. The first figure corresponds to the cost of $f$ on these nodes. The following two figures show the cost of $g$ on these nodes, assuming that they belong to different clusters $A$ and $B$. Fourth figure shows the set of active nodes in the tree (colored in green) after $T(v)$ gets clustered. Fifth figure presents the setting in Proposition \ref{prop:mla_light_cluster_costs}.}
\label{fig:figure_one}
\end{figure}

First, we present the idea behind our approach. 
Recall that our goal is to find a partition $P_1, \ldots, P_k$ of nodes of $T$ such that the gap between the values $f(V)$ and $g(V)$ is at most of order $\sqrt{n}$ for all the sets $V \subseteq U$. 
Here, $f(V)$ is the cost of minimal subtree $T_V$ of $T$ that connects all the nodes in $V$ to the root, as stated before. 
On the other hand, $g(V)$ needs to cover the costs of all the parts in $P$ that intersect with $V$. 
For instance, if $V$ intersects exactly two parts $A$ and $B$ in $P$, then $g(V) = f(A) + f(B)$. 
Although these parts themselves are disjoint by the definition of partition $P$, as we pay for each of them separately in $g$ (by paying for set $A$, we mean generating the cost of $f(A)$), we incur not only the costs of their nodes $c(A)$ and $c(B)$ but also the costs of the paths that connect them to the root $r$ (see Figure \ref{fig:figure_one}).

Note that this process can cause us to incur two types of additional costs with respect to the optimal value $f(V)$. 
First, both parts $A$ and $B$ may contain not only the nodes in $V$ but also their neighbors, for which we need to pay as well. 
Second, as we pay for each part separately, we may be forced to pay for some nodes on the paths to the root multiple times (see Figure \ref{fig:figure_one}). 

Since $f(V)$ is equal to the cost of the nodes in $T_V$ and we aim for $g$ to be $\sqrt{n}$-approximation of $f$, the intuition is that $g$ can afford to pay the cost of each node in $T_V$ roughly $\sqrt{n}$ times (as this gives the desired ratio). 
This observation provides the foundations for our algorithm. 
Let us remark that at the beginning, all the nodes in $T$ are unpartitioned, i.e., $P = \emptyset$. 
Our algorithm revolves around two procedures. 
The first one can be seen as assigning each node $v$ in $T$ a budget of $\alpha \sqrt{n} \cdot c(v)$ for some $\alpha \in \mathbb{N}$. 
A vertex $v$ may then use such a budget to create a new part $K$ in the partition. 
We allow $v$ to generate only one form of a \emph{cluster}, i.e., a part to be included in $P$, that consists of all the unpartitioned nodes in its subtree $T(v)$. 
Furthermore, for such a part $K$ to be added to $P$, it needs to hold that the costs of (the unpartitioned nodes in) $T(v)$ and $R(v)$ both fit into $v$'s budget. 
If we manage to add $K$ to $P$, we call both node $v$ and cluster $K$ \emph{heavy}.

Whenever the first procedure cannot be applied, i.e., there are no vertices that can generate heavy clusters from the unpartitioned nodes, we run the second procedure. 
The idea then is to find a subtree (or a family of subtrees) of size roughly $\sqrt{n}$ (details to be presented later) and group them together into a new part in the partition. 
We call this part a \emph{light cluster}. 
In case there are nodes that become heavy after this action (as their descendants got clustered), we go back to the first procedure, which starts a new iteration of the main algorithm. 

Notice that the idea behind the second procedure is to upper bound the number of times we need to pay the cost of the paths connecting the clusters to the root $r$. 
Since $T$ has $n$ nodes and each light cluster is of size close to $\sqrt{n}$, we can only create roughly $\sqrt{n}$ such clusters. 
Thus, even when $V$ intersects all the light clusters, we pay for the nodes in $T_V$ at most $O(\sqrt{n})$ times, which we can afford. 
It remains to estimate the cluster costs, which follow in the next section.

\subsection{MLA Partitioning Algorithm}

\paragraph*{Heavy clusters.} Let us first present two definitions. Here, we assume that whenever we are given a partially created partition $\tilde{P}$ of nodes of an MLA tree $T$, then the set of nodes it already partitioned, i.e., $V(\tilde{P}) = \bigcup \tilde{P}$, does not disconnect $T$, i.e., tree $T' = T \setminus \bigcup \tilde{P}$ is a subtree of $T$.

\begin{definition} \label{def:active_nodes}
    Let $T$ be a tree given in an MLA instance, denote the set of its nodes by $U$, and let $\tilde{P}$ be a partially created partition of $U$ (i.e., $V(\tilde{P}) = \bigcup \tilde{P}$ is a proper subset of $U$). 
    Then we call all the nodes that are not partitioned yet, i.e., belong to $U \setminus V(\tilde{P})$, \emph{active} (see Figure \ref{fig:figure_one}). 
    We use the notation of $V|_{act}$ to restrict any subset $V$ of nodes in $T$ to the nodes that are active.
\end{definition}

\begin{definition} \label{def:heavy_node_and_cluster}
    Let $(T, c)$ be an MLA instance, and let $\tilde{P}$ be a partially created partition of $T$'s nodes. 
    We say that an active node $v$ is \emph{heavy} if the costs of path $R(v)$ and subtree $\act{T}(v)$ are at most $4\sqrt{n} \cdot c(v)$ each. 
    If we extend $\tilde{P}$ by adding $\act{T}(v)$, we call this new part a \emph{heavy} cluster.
\end{definition}

Now, we can prove a simple fact about heavy clusters.

\begin{propositionapprep} \label{prop:heavy_clusters_cost}
    Let $(T, c)$ be an MLA instance. Take any partition $P$ of nodes of $T$ and let $P_{h,1}$, $\ldots$, $P_{h,s}$ be a sublist of all heavy clusters in $P$. We denote their roots by $v_{h,1}$, $\ldots$, $v_{h,s}$, respectively, and the set containing them by $V_h$. Then, it holds that $\sum_{i=1}^s f(P_{h,i}) \leq 8\sqrt{n} \cdot f(V_h)$.
\end{propositionapprep}
\begin{proof}
    By Definition \ref{def:heavy_node_and_cluster}, we have that for each node $v_{h,i}$ the following is satisfied: $c(P_{h,i}) \leq 4\sqrt{n} \cdot c(v_{h,i})$ and $c(R(v_{h,i})) \leq 4\sqrt{n} \cdot c(v_{h,i})$. The first inequality here comes from the fact that cluster $P_{h,i}$ was the set of all active nodes (Definition \ref{def:active_nodes}) contained in the subtree $T(v_{h,i})$ at the moment it was created (i.e., it was $\act{T}(v_{h,i}))$. Hence, it holds that 
    \begin{align} \label{ineq:heavy_cluster}
    \begin{split}
        f(P_{h,i}) & = c(P_{h,i}) + c(\excendp{R}(v)) \leq c(P_{h,i}) + c(R(v)) \\[2pt]
        & \leq 4\sqrt{n} \cdot c(v_{h,i}) + 4\sqrt{n} \cdot c(v_{h,i}) = 8\sqrt{n} \cdot c(v_{h,i}),
    \end{split}
    \end{align}
    where the first equality comes from the fact that $P_{h,i}$ is a subtree, which means that the minimal tree containing all its nodes and the root $r$ is only missing the path from $v$ to $r$ (with $v$ excluded as we already counted it in the cluster). 
    Moreover, let us notice that $f(V_h) \geq \sum_{i=1}^s c(v_{h,i})$, as it is the cost of the minimal tree containing all the nodes $v_{h,i}$. 
    Thus, to obtain the desired inequality, we only need to sum \eqref{ineq:heavy_cluster} over all the heavy clusters and then apply the inequality above.
\end{proof}

\paragraph*{Light clusters.} Here, we present a procedure that generates a light cluster. 

\begin{definition} \label{def:light_cluster}
    Let $(T, c)$ be an MLA instance, and let $\tilde{P}$ be a partially created partition of $T$'s nodes. 
    We say that a subset $K$ of nodes of $\act{T}$ is a \emph{light} cluster if (1) its size fits into the range $I(n) := [\sqrt{n}, 2\sqrt{n}]$, (2) it is either a maximal subtree in $\act{T}$ or a collection of maximal subtrees having the same parent, and (3) $\act{T}$ does not contain any heavy nodes.\footnote{This third condition is for analysis purposes only and the property giving the name to \emph{light} clusters.} 
    In case $\act{T}$ is of size smaller than $\sqrt{n}$, and we set $K = \act{T}$, we drop the first condition and still call $K$ a light cluster.
\end{definition}

\noindent
Given the definition above, we present Algorithm \ref{alg:mla_light_cluster} that shows how to find such a cluster.

\begin{algorithm}[htbp]
\caption{MLA Light Cluster Search}
\label{alg:mla_light_cluster}

\textbf{Input}: MLA tree $T$ with some nodes marked active ($\act{T}$ is a subtree of $T$ containing its root $r$) \\
\textbf{Output}: light cluster formed of the nodes in $\act{T}$

\begin{algorithmic}[1]

\STATE \textbf{if} $|\act{T}| \leq 2\sqrt{n}$ \label{line:mla_cluster_whole_tree}

\STATE \ \ \ \textbf{return} $\act{T}$   \label{line:mla_subtree_with_r}

\STATE $u := r$   \label{line:mla_start_at_the_root}

\STATE \textbf{while} there exist a node $v \in C_{act}(u)$ such that $|T_{act}(v)| > 2\sqrt{n}$   \label{line:mla_size_child_two_sqrt}

\STATE \ \ \ $u := v$   \label{line:mla_set_child_as_the_current_node}

\STATE \textbf{if} there exist a node $v \in C_{act}(u)$ such that $|T_{act}(v)| \geq \sqrt{n}$ \textbf{then}   \label{line:mla_size_child_sqrt}

\STATE \ \ \ \textbf{return} $T_{act}(v)$   \label{line:mla_child_node_subtree}

\STATE \textbf{else}   \label{line:mla_needs_to_cluster_forest}

\STATE \ \ \ denote all the elements in $C_{act}(u)$ by $v_1, v_2, \ldots, v_j$ for some $j$ \label{line:mla_sort_children_by_decreasing_weights}

\STATE \ \ \ set iterator $i = 1$ and initialize a new cluster $V$ with an empty set

\STATE \ \ \ \textbf{while} $|V| < \sqrt{n}$   \label{line:mla_forest_loop}

\STATE \ \ \ \ \ \ add $T_{act}(v_i)$ to $V$

\STATE \ \ \ \ \ \ increment $i$ by 1

\STATE \ \ \ \textbf{return} $V$   \label{line:mla_child_nodes_forest}

\end{algorithmic}
\end{algorithm}

\begin{propositionapprep}
    If there are no heavy nodes in $\act{T}$ (see condition (3) in Definition \ref{def:light_cluster}), then Algorithm \ref{alg:mla_light_cluster} finds a light cluster in $\act{T}$. 
\end{propositionapprep}
\begin{proof}
    Notice that we start the search of a new cluster by checking whether the size of $\act{T}$ (the subtree containing all the active nodes in $T$) is smaller or equal to $2\sqrt{n}$ (line \ref{line:mla_cluster_whole_tree}). 
    If so, we return the whole tree $\act{T}$ since it fits into the description given in the last sentence of Definition \ref{def:light_cluster}. 
    Otherwise, we set $r$ to be the current node we are at, which we denote by $u$ (line \ref{line:mla_start_at_the_root}). 
    Then, we go through the while loop from line \ref{line:mla_size_child_two_sqrt} to \ref{line:mla_set_child_as_the_current_node}, each time picking a child $v$ of the current node $u$ such that the subtree $\act{T}(v)$ is of size greater than $2\sqrt{n}$. 
    If such a node exists, we move to it, setting $u = v$, and we leave the while loop otherwise.
    
    In the second case, we know that, as we go to line \ref{line:mla_size_child_sqrt}, two conditions are satisfied. 
    First, the size of the subtree $\act{T}(u)$ rooted at the current node $u$ is at least $2\sqrt{n}$. 
    Indeed, we either stayed at the root node, not satisfying the condition in the if statement in line \ref{line:mla_cluster_whole_tree}, or we further went from $r$ through a sequence of its descendants, each having a subtree of size greater than $2\sqrt{n}$. 
    Second, none of $u$'s children has a subtree of size greater than $2\sqrt{n}$, as we already left the while loop. 
    
    Now, in line \ref{line:mla_size_child_sqrt}, we check whether there exists a child $v$ of the current node, which subtree $\act{T}(v)$ is of size at least $\sqrt{n}$. 
    If so, we return $\act{T}(v)$, as it satisfies the conditions to be a light cluster. 
    Otherwise (line \ref{line:mla_needs_to_cluster_forest}), we iterate through $u$'s children $v_i$ (line \ref{line:mla_forest_loop}) and add the nodes contained in their subtrees $\act{T}(v_i)$ to a set $V$. 
    We stop at the moment when the size of $V$ becomes at least $\sqrt{n}$ and return $V$ as a new cluster. 
    It is easy to notice that in the while loop, we indeed need to pass the $\sqrt{n}$ size threshold, as $|\act{T}(u)| > 2\sqrt{n}$.
    Moreover, we know that before we added the nodes of the last subtree $T'$ to $V$, $V$ had a size smaller than $\sqrt{n}$. 
    Since $|T'| < \sqrt{n}$, we have that the whole group is of size smaller than $2\sqrt{n}$.
\end{proof}

\paragraph*{Main algorithm.} 

Before we describe the partitioning algorithm, let us introduce a helper function. 
We define method $\texttt{cluster}(V)$ to group all the elements of $V$ together and include them as a new part in the partition. 
Let us also emphasize that after this call, all the elements in $V$ become inactive. 
With the above notation, we can formalize our approach as presented in Algorithm \ref{alg:mla_partition}.

\begin{algorithm}[htbp]
\caption{\label{alg:mla_partition} MLA Partitioning Algorithm}

\textbf{Input}: MLA instance $(T,c)$ \\
\textbf{Output}: partition $P$ of the nodes of $T$

\begin{algorithmic}[1]

\STATE initialize an empty partition $P$

\STATE \textbf{while} $\act{T}$ is not empty

\STATE \ \ \ \textbf{while} there exist a heavy node $v \in \act{T}$   \label{line:mla_heavy_node_exists_check}

\STATE \ \ \ \ \ \ $\texttt{cluster}(\act{T}(v))$   \label{line:mla_heavy_node_subtree}

\STATE \ \ \ \textbf{if} $\act{T}$ is empty   \label{line:mla_empty_tree_after_heavy_cluster_check}

\STATE \ \ \ \ \ \ \textbf{break}

\STATE \ \ \ apply Algorithm \ref{alg:mla_light_cluster} to find a light cluster $V$ in $\act{T}$   \label{line:mla_light_cluster_search}

\STATE \ \ \ $\texttt{cluster}(V)$   \label{line:mla_light_cluster_creation}

\STATE \textbf{return} $P$

\end{algorithmic}
\end{algorithm}

As mentioned in the first part of this section, the main partitioning algorithm runs heavy and light cluster searches in a loop. 
In the first step, it iteratively finds the heavy clusters in the tree $\act{T}$ determined by the already created partition (lines \ref{line:mla_heavy_node_exists_check} to \ref{line:mla_heavy_node_subtree}). 
Then, if tree $T$ is not yet partitioned (condition in line \ref{line:mla_empty_tree_after_heavy_cluster_check} does not hold),  it goes to the second step that finds one light cluster and adds it to the partition (lines \ref{line:mla_light_cluster_search} to \ref{line:mla_light_cluster_creation}). 
After this point, it goes to the initial step and loops.

Let us emphasize that during the whole partitioning procedure, the set $\act{T}$ of all active elements in $T$ forms a proper subtree containing the root $r$ of $T$. 
Indeed, in the beginning, $T_{act} = T$ and all the \texttt{cluster} calls truncate one or more maximal subtrees from $T_{act}$. 
Now, given Algorithm \ref{alg:mla_partition}, we go back to proving the properties of light clusters. 

\begin{propositionapprep}
\label{prop:mla_light_clusters_numer}
    Let $T$ be an MLA tree rooted at some node $r$ and let $P$ be the partition of nodes of $T$ created by Algorithm \ref{alg:mla_partition}. 
    We denote all the light clusters in $P$ by $P_{\ell,1}, P_{\ell,2}, \ldots, P_{\ell,t}$ and require them to be listed in the creation order. 
    Then, it holds that there are at most $\sqrt{n} + 1$ parts $P_{\ell,i}$.
\end{propositionapprep}
\begin{proof}
    Notice that by the definition, the only light cluster that can have a size smaller than $\sqrt{n}$ is the one containing the root $r$. 
    Thus, all the light clusters created before, i.e., at least $t - 1$ of them, have the size at least $\sqrt{n}$.
    Since there are $n$ nodes in tree $T$, we get that there are at most $n/\sqrt{n} = \sqrt{n}$ such clusters. Thus, $t \leq \sqrt{n} + 1$, which concludes the proof.
\end{proof}

In the remaining part of this section, we refer to the clusters created in lines \ref{line:mla_subtree_with_r}, \ref{line:mla_child_node_subtree} of Algorithm \ref{alg:mla_partition}, i.e., the ones that consist of a single subtree, as the light clusters of type I. 
We call the light clusters consisting of forests (created in line \ref{line:mla_child_nodes_forest}) the light clusters of type II. 
We prove that the cost function $c$ satisfies the following properties. Here, we overuse the notation of $c$ and extend it to the subsets as well, i.e., for any $V \subseteq U$ we set $c(V) = \sum_{v \in V} c(v)$.

\begin{propositionapprep}
\label{prop:mla_light_cluster_costs}
    Let $(T, c)$ be an MLA instance and let $P$ be the partition obtained on it by Algorithm \ref{alg:mla_partition}. Take any light cluster $K$ in $P$ and denote by $r_K$ the root of $K$ if it is a cluster of type I. Otherwise, if $K$ is a cluster of type II, we use $r_K$ to denote the parent node of the forest contained in $K$. Then it holds that $c(P(r_K)) \geq c(K)$.
\end{propositionapprep}
\begin{proof}
    Without loss of generality, assume that $K$ is of type I. 
    Let $w$ be the node in $K$ that has the highest cost. 
    By Definition \ref{def:light_cluster}, we know that $|K| \leq 2\sqrt{n}$. 
    Hence, by an averaging argument, we have $c(w) \geq c(K) / (2\sqrt{n})$, which implies $2\sqrt{n} \cdot c(w) \geq c(K)$. 
    Now, assume by contradiction that $c(P(r(K))) < c(K)$. 
    Then, if we split the path from $w$ to $r$ into two parts by cutting it on the node $r_K$, we got $c(P(w)) = c(P(w) \cap K) + c(\excendp{P}(r_K)$. 
    Since $c(P(w) \cap K) \leq c(K)$ and $c(P(r_K) \leq c(K)$ by our assumption, we get that $c(P(w)) \leq 2c(K) \leq 2 \cdot 2\sqrt{n} \cdot c(w) = 4\sqrt{n} \cdot c(w)$. 
    However, this means that $w$ is a heavy node, which contradicts the initial assumption. 
    Thus, it holds that $c(P(r_K)) \geq c(K)$. 
    The proof for type II follows the same steps.
\end{proof}

\begin{corollaryapprep}
    Let us subsume the notation and the conditions of Proposition \ref{prop:mla_light_cluster_costs}. Then, it holds that $f(K) \leq 2f(r_K)$.
\end{corollaryapprep}
\begin{proof}
    Notice that for type I cluster $K$, $f(K)$ consists of the cost of $K$ and the cost of the path connecting it to the root $r$ of $T$ (to be precise, excluding $r_K$ from this path, as we already count its cost in the cluster). Thus, the following holds
    \[
        f(K) = c(K) + c(\excendp{P}(r_K)) \leq c(P(r_K)) + c(\excendp{P}(r_K)) \leq 2 c(P(r_K)) = 2 f(r_K),
    \]
    where the first inequality is implied by Proposition \ref{prop:mla_light_cluster_costs}, the second one from the fact that we added the cost of $r_K$ to the right side, and the last inequality is by the definition of $f$.
\end{proof}
    
Given the above, we can prove the main theorem of this section.

\begin{theorem}\label{thm:mla}
    For any MLA service function $f$, there exists a disjoint service function $g$ that $O(\sqrt{n})$-approximates $f$.
    It can be found in time polynomial w.r.t.\ the MLA instance defining $f$.
\end{theorem}
\begin{proof}
    Let $(T,c)$ be the MLA instance that defines $f$, and let $U$ be the set of nodes in $T$.
    The idea is to prove that the partition $P = \{P_1, P_2, \ldots, P_k\}$ generated on $T$ by Algorithm \ref{alg:mla_partition} induces a set function $g(V)=\sum_{i\in [k]} f(P_i) \ind{P_i \cap V \neq \emptyset}$ on all subsets $V \subseteq U$ that is an $O(\sqrt{n})$-approximation to $f$.
    The function $g$ is a disjoint service function by design.
    
    For this purpose, we need to show that $\max_{V \subseteq U} g(V) / f(V)$ is of order at most $\sqrt{n}$. 
    Let us note that in our case, $f(V)$ is just the cost of the minimal subtree connecting $V$ to the root. Thus, for any subset $V'$ of $V$ it holds that $f(V') \leq f(V)$.

    Let $V \subseteq U$ be any subset of nodes and let $P_{h,1}$, $\ldots$, $P_{h,s}$ and $P_{\ell,1}, P_{\ell,2}, \ldots, P_{\ell,t}$ be the lists of all the heavy and light clusters that intersect $V$, respectively. 
    We also denote the roots of the heavy clusters by $v_{h,1}$, $\ldots$, $v_{h,s}$, respectively, and the set containing them by $V_h$. 
    Similarly, we use the convention from Proposition \ref{prop:mla_light_cluster_costs} to define light cluster nodes. For $P_{\ell,i}$, we denote its root by $r_{\ell,i}$.
    
    By Proposition \ref{prop:heavy_clusters_cost}, it holds that $\sum_{i=1}^s f(P_{h,i}) \leq 8\sqrt{n} \cdot f(V_h)$. 
    Moreover, since $V$ intersects all these heavy clusters, it either contains their roots or some nodes that are their descendants. 
    Thus, the minimal subtree connecting $V$ to the root $r$ contains the minimal subtree connecting $V_h$ to the root $r$. 
    Hence,
    \begin{equation} \label{eq:estimation_on_heavy}
        \sum_{i=1}^s f(P_{h,i}) \leq 8\sqrt{n} \cdot f(V_h) \leq 8\sqrt{n} \cdot f(V).
    \end{equation}

    Now, for each light cluster $P_{\ell,i}$, we notice that since $V$ intersects it, the minimal tree connecting $V$ to $r$ contains the path from $r_{\ell,i}$ to $r$. 
    Thus, $f(V) \geq f(r_{\ell,i})$ and by Proposition \ref{prop:heavy_clusters_cost}, we get that
    \begin{equation} \label{eq:estimation_on_light}
        f(P_{\ell,i}) \leq 2f(r_K) \leq 2f(V)
    \end{equation}
    for each $\ell \in [t]$. Note that $g(V) = \sum_{K \in P: V \cap K \not= \emptyset} f(K)$. 
    Combining inequalities \ref{eq:estimation_on_heavy} and \ref{eq:estimation_on_light}, we obtain that 
    \begin{equation*}
    \begin{split}
       \frac{g(V)}{f(V)}  = \frac{\sum_{K \in P: V \cap K \not= \emptyset} f(K)}{f(V)} & = \frac{\sum_{i=1}^s f(P_{h,i}) + \sum_{i=1}^t f(P_{\ell,i})}{f(V)} \leq \frac{8\sqrt{n} \cdot f(V) + \sum_{i=1}^t 2f(V)}{f(V)} \\
        & \leq \frac{8\sqrt{n} \cdot f(V) + 2(\sqrt{n} + 1) \cdot f(V)}{f(V)} = 10\sqrt{n} + 2,
    \end{split}
    \end{equation*}
    with the last inequality implied by Proposition \ref{prop:mla_light_clusters_numer}.
    This concludes the proof that $g$ is an $O(\sqrt{n})$-approximation to $f$.

    Finally, it is easy to notice that the algorithm runs in polynomial time. We can define a dynamic structure over the tree $T$ that, for each node $v$, stores its subtree and path costs ($c(T(v))$, $c(P(v))$, together with the size $|T(v)|$ of its subtree. Updates on such a structure take at most polynomial time in $n$ (as we create a cluster, we go from the cluster root to the root of $T$, updating the data on all the nodes on the path, which is of length at most $n$). With such a structure, checking whether a node is heavy or going through a path from $r$ in search of a light cluster also takes at most linear time in $n$.
\end{proof}

Thus, by Lemma~\ref{lem:reduction-to-disjoint}, we get Theorem~\ref{thm:non-clairvoyant-mla}. The result of Theorem~\ref{thm:mla} is tight:

\begin{restatable}{proposition}
{proptightmla}
    \label{prop:tightmla}
    There exists a decreasing MLA instance $T,c$ with $n$ nodes, such that for every partition $P_1,\ldots,P_k$ of $T$ for some $k$, there exists a non-empty set $S\subseteq T$ such that $$\frac{\sum_{i\in [k]} f(P_i) \ind{S\cap P_i\neq \emptyset }}{f(S)}= \Omega(\sqrt{n}) .$$
\end{restatable}

\begin{proof}
    Consider the tree $T$ with a root $r$ and $n-1$ children of $r$ denoted by $v_1,\ldots,v_{n-1}$.
    The cost $c$ is such that $c(r)=\sqrt{n}$, while for all $i\in [n-1]$, $c(v_i)=1$.
    Now, consider any partition $P_1,\ldots,P_k$.
    If $k>\sqrt{n}$, then consider a set $S$ that intersects each $P_i$ exactly once. Thus, $f(S) \leq \sqrt{n} + k \leq 2k$, while $\sum_{i\in [k]} f(P_i) \ind{S\cap P_i\neq \emptyset } \geq k \cdot \sqrt{n}$, which proves this case.
    Else ($k \leq \sqrt{n}$), consider a set $S$ that intersects $P_i$ once if and only if $|P_i| < \sqrt{n}/2 $ (otherwise does not intersect at all).
    It holds that $f(S) \leq 2\sqrt{n}$ because one has $\sqrt{n}$ from $r$ and $\sqrt{n}$ intersections with $P_i$'s in the worst case,
while  $$\sum_{i\in [k]} f(P_i) \ind{S\cap P_i\neq \emptyset } \geq n-\sum_{i\in [k]} |P_i|\cdot \ind{S\cap P_i =  \emptyset } \geq n- k \sqrt{n}/2 \geq n/2, $$
    which concludes the proof.
\end{proof} 
\section{Weighted Symmetric Subadditive Joint Replenishment} \label{sec:weighted}
In this section, we study Weighted Symmetric Subadditive JRP. We have a set $U$ of $n$ request types with weights $w(\{j\}) = w_j$ for each $j \in U$. Let $f$ be the set function over $U$: In this setting, we have that the service cost of a set $S$ only depends on the total weight of the elements belonging to $S$, as opposed to the identity of those elements. Formally, $f(S) = f(w(S))$, where function $f$ is now intended as a monotone non-decreasing subadditive function of weights of a set with $f(0) = 0$, and for every two weights $x,y$, it holds that $f(x+y)\leq f(x)+f(y)$. For brevity, we call these functions \textit{weighted symmetric subadditive}. Our goal is to show that for every weighted symmetric subadditive service function $f$ on $U$, there exists a partition of $U$ into sets $S_1,\ldots,S_k$ for some $k$, such that the disjoint service function $g: U \rightarrow \R_{\geq 0}$ defined by this partition where $g(S) = \sum_{i=1}^k f(S_i)\cdot\indicator{S\cap S_i \neq \emptyset}$ satisfies $g(S) \leq O(\sqrt{n})f(S)$  for every $S \subseteq U$.

We begin, in Section~\ref{sec:cardinality}, by analyzing a special case of unweighted symmetric subadditive service costs. Namely, where the weight of each element is $1$, and thus, $w(S) = |S|$: These functions are simply referred to as \textit{symmetric subadditive}. We achieve a tight $\Theta(\sqrt{n})$-stretch with a simple partitioning algorithm (partition into $\sqrt{n}$ sets of size $\sqrt{n}$ each), and this serves as a warm-up to the weighted symmetric subadditive case presented in Section~\ref{sec:weighted-concave}, where we also achieve a tight $\Theta(\sqrt{n})$-stretch. 

\subsection{Symmetric Subadditive JRP}
\label{sec:cardinality}
We first consider symmetric subadditive service functions. Observe that these functions are symmetric (i.e., $f(S) = f(S')$ for all sets $S,S' \subseteq U$ such that $|S| = |S'|$). For convenience, for a cardinality $0\leq s\leq |U|$, we use $f(s)$ as the value of sets of size $s$. 
We show that for symmetric subadditive $f$, one can construct  a disjoint service function $g$ that $O(\sqrt{n})$-approximates $f$. We then show that $O(\sqrt{n})$ is tight even in the special case of $f$ being a symmetric unweighted set cover function.
This provides an alternative, simpler proof for the lower bound on USC of \cite{JiaLNRS05}. We first state the following simple but useful observation.

\begin{restatable}{observation}{obssubadd}
    \label{obs:subadd}
    For all symmetric subadditive functions $f: \R_{+} \rightarrow \R_{+}$, and all $y \geq x > 0$, it holds that $f(y)/f(x) \leq \lceil y/x \rceil$. 
\end{restatable}

\begin{proof}
    Let $k = \lceil y/x \rceil$. Then,
    $f(y) \leq f(kx) 
 \leq  f(x) + f((k-1)x) \leq f(x) +\ldots +f(x) = k \cdot f(x)$.
 The first inequality is by monotonicity, and the second and third by subadditivity.   
\end{proof}

\begin{restatable}{lemma}{lemsym}
    \label{lem:sym}
       For every symmetric subadditive service function $f$, there exists a disjoint service function $g$ that $O(\sqrt{n})$-approximates it. 
\end{restatable}

\begin{proof}
    Let us consider an arbitrary symmetric subadditive service function $f$ on request types $U$. Let $g$ be the disjoint service function that induces an arbitrary partition of the elements of $U$ into sets $\{X_1, \ldots, X_k\}$, where $k = \lceil\sqrt{n}\rceil$, each of cardinality $|X_i| \leq \lceil \sqrt{n}\rceil$ (such a partition always exists). We now bound the following fraction for every $S \subseteq U$:
    \begin{align*}
        \frac{\sum_{i \in [k]} f(X_i) \cdot \ind{S \cap X_i \neq \emptyset}}{f(S)} &\leq \frac{\sum_{i \in [k]} f(\lceil \sqrt{n} \rceil ) \cdot \ind{S \cap X_i \neq \emptyset}}{f\left(\sum_{i \in [k]} \ind{S \cap X_i \neq \emptyset}\right)} \\
        & \leq \left(\sum_{i \in [k]} \ind{S \cap X_i \neq \emptyset}\right) \cdot \left\lceil\frac{\lceil \sqrt{n} \rceil }{\sum_{i \in [k]} \ind{S \cap X_i \neq \emptyset}}\right\rceil \\
        &\leq 2\lceil \sqrt{n} \rceil.
    \end{align*}
The first inequality is because $|X_i| \leq \lceil \sqrt{n}\rceil$ and from the fact that, since $X_i$'s are disjoint, the size of $S$ is at least the number of non-empty intersections with sets $X_i$'s. The second inequality follows from Observation~\ref{obs:subadd}, and the third inequality follows since $\lceil\frac{a}{b}\rceil \leq 2 \cdot \frac{a}{b}$, for every $\frac{a}{b} \geq \frac{1}{2}$, and the denominator $\sum_{i \in [k]} \ind{S \cap X_i \neq \emptyset} \leq \sqrt{n} + 1$. \end{proof}

Thus, by Lemma~\ref{lem:reduction-to-disjoint} and Lemma~\ref{lem:sym}, the following holds:
\begin{theorem}
    There exists a deterministic $O(\sqrt{n})$-competitive algorithm for the Non-Clairvoyant Symmetric Subadditive Joint Replenishment problem. 
\end{theorem}

\noindent We complement the above result by giving a tight instance:

\begin{restatable}{theorem}{thmlb}
    \label{thm:lb}
There exists a symmetric subadditive service function such that every disjoint service function is an $\Omega(\sqrt{n})$-approximation of it.
\end{restatable}

\begin{proof}
    Let $U$ be the set of request types. For simplicity of the proof, we assume that $n=|U|$ has an integer square root. Let us consider the service function $f(S) = \clfrac{|S|}{\sqrt{n}}$, which is symmetric and subadditive, let $g$ be any disjoint service function, let $\cS$ be the collection of disjoint sets $X_i$'s that $g$ generates, and let $k$ be the number of parts in the partition $\cS$. 

    Consider some $X \subseteq U$ that intersects each $X_i$ exactly once.  We now analyze the cost of this induced partition on $X$:
    \[
        \frac{\sum_{i=1}^k f(X_i)}{f(X)} =  \frac{\sum_{i=1}^k \clfrac{|X_i|}{\sqrt{n}}}{\clfrac{k}{\sqrt{n}}} \geq \frac{\max\{k,\sqrt{n}\}}{\clfrac{k}{\sqrt{n}}} ,
    \]
    where the inequality holds since the it is a sum of $k$ terms where each is at least $1$, and since the sets $X_1,\ldots,X_k$ cover $U$, thus $\sum_{i=1}^k |X_i|  = n$.
    
    Now, if $k \leq \sqrt{n}$ then $$ \frac{\max\{k,\sqrt{n}\}}{\clfrac{k}{\sqrt{n}}} = \sqrt{n}.$$
    Otherwise, $\frac{k}{\sqrt{n}} >1$, thus $\clfrac{k}{\sqrt{n}} \leq 2 \frac{k}{\sqrt{n}}$, which implies that $$ \frac{\max\{k,\sqrt{n}\}}{\clfrac{k}{\sqrt{n}}} \geq  \frac{k}{2k/\sqrt{n}} = \frac{\sqrt{n}}{2},$$
    which concludes the proof.
\end{proof}

 \subsection{Weighted Symmetric Subadditive JRP}\label{sec:weighted-concave}

We now relax the assumption of $w(S) = |S|$ and provide a $O(\sqrt{n})$-approximation for every weighted subadditive function.

We begin with some facts about weighted subadditive and symmetric concave functions. Every symmetric concave function is the pointwise infimum of a set of affine functions, and can be approximated by a set of affine functions with exponentially decreasing slopes. The next lemma combines this fact with the fact that every weighted subadditive function can be approximated by a symmetric concave function. 
\begin{lemmaapprep}
    \label{lem:piecewise}
    Let $g:\{0,1\ldots,W\}\rightarrow \R_{\geq 0 }$ be a monotone non-decreasing subadditive function. Then, there exists a finite set of affine functions $\{g_1,\ldots,g_p\}$ for some $p \leq \log(W)$ where $g_i(x) = \sigma_i + x\cdot \delta_i$ such that $\sigma_{i+1} > 2\sigma_i $ and $\delta_{i+1}< \delta_i / 2$ for every $i<p$, and the function $\hat{g}$ defined by $\hat{g}(x) = \min_i g_i(x)$  satisfies that for every $x\in\{0,\ldots,W\}$, it holds that: $$g(x) \leq \hat{g}(x) \leq 8 g(x).$$
\end{lemmaapprep}
\begin{proof}
By \cite{EzraFRS20}, we know that there exists a concave function $g^\prime:\{0,\ldots,W\}\rightarrow \R_{\geq 0}$ that approximates $g$ within a factor of $2$.
Now, for every $i=2,\ldots,\lceil\log(W)\rceil$ consider the affine function $g^\prime_i : \{0,\ldots,W\}\rightarrow \R_{\geq 0}$  that interpolates between $(2^{i-1},g^\prime(2^{i-1}))$ and $(2^{i},g^\prime(2^{i}))$, and $g_1'(x)$ that interpolates between $(0,g^\prime(0))$ and $(1,g^\prime(1))$. It holds that for every $x \in \{0,\ldots,W\}$ then $$\frac{g^\prime(x)}{2} \leq \min_{i=1,\ldots,p} g^\prime_i(x) \leq g^\prime(x),  $$ where the first inequality holds since
\begin{align*}
    g^\prime(x) &\leq g^\prime(2^{\lceil \log(x)\rceil}) \leq 2g^\prime(2^{\lfloor \log(x)\rfloor})   \\
    &\leq 2g^\prime_{2^{\lfloor \log(x)\rfloor}}(2^{\lfloor \log(x)\rfloor}) =  2\min_{i=1,\ldots,p} g^\prime_i(2^{\lfloor \log(x)\rfloor})  \leq 2\min_{i=1,\ldots,p} g^\prime_i(x),
\end{align*}
and the second inequality holds by concavity of $g^\prime$.
In \cite{GuhaMM09}, they present an algorithm that reduces the set of affine functions such that the coefficients and slopes satisfy the conditions of the lemma while losing a factor of $2$, which, if applied to the set of affine functions $2g^\prime_i$, concludes the proof. 
\end{proof}

Henceforth, we will denote by $W=w(U)$, and assume that $f$ is defined on $\{0,\ldots,W\}$, and is a pointwise infimum of {$p$ affine functions $g_1, \ldots, g_p$} where $g_i(x) = \sigma_i + x \cdot \delta_i$ and the $\sigma_i$s and $\delta_i$s satisfy the properties stated in the lemma. Proving the theorem for $f$ that satisfies the condition proves the same (with additional loss of a factor of $8$) for general symmetric subadditive functions. 

The following lemma will be useful to lower bound $f(w(S))$ using the largest weight in $S$.

\begin{lemma}
    \label{lem:crossover}
    For every $k \in \{2,\ldots,p\} $, if  $x \geq \frac{\sigma_k}{\delta_{k-1}}$ , then $f(x) \geq \sigma_k$. \end{lemma}

\begin{proof}
Recall that $f(x) = \min_{1 \leq i \leq p} \sigma_i + x\delta_i$. For $i < k$, we have $\sigma_i + x\delta_i\geq x \delta_{k-1} \geq \sigma_k$. For $i \geq k$, we have $\sigma_i + x \delta_i \geq \sigma_k$. Thus, $f(x) \geq \sigma_k$.
\end{proof}

Henceforth, for brevity, we write $f(S)$ to mean $f(w(S))$, for an arbitrary set $S$. In the following, we frequently use the fact that for any set $H$, $f(H) = \min_{1 \leq i' \leq p} \sigma_{i'} + w(S)\delta_{i'} \leq \sigma_i + w(H)\delta_i$ for every $i$.

\paragraph{High-Level Overview.} 
Let $S$ be a set chosen by an adversary, unknown to us. Suppose that $f(S) = \min_{1 \leq i \leq p} \sigma_i + w(S)\delta_i = \sigma_\ell + w(S)\delta_\ell$.
The idea is to construct a partition such that some of the parts that intersect $S$ can be charged to $\sigma_\ell$, and the remaining parts that intersect $S$ can be charged to $w(S)\delta_\ell$. Towards this end, we first classify each type $j$ as follows. We say that type $j$ is \emph{eligible} for \emph{class} $2 \leq k \leq p$ if $w_j \geq \frac{\sigma_k}{\delta_{k-1}}$. All types are eligible for class $1$. Define the \emph{class} of type $j$ to be the largest class it is eligible for and $X_k$ to be the set of class-$k$ types. 

Next, we partition $X_k$ into heavy and light types. The light part $Z_k$ contains all types $j \in X_k$ with $w_j\delta_k \leq \sigma_k/\sqrt{n}$. Since $Z_k$ is light, $f(Z_k) \leq \sigma_k + w(Z_k) \delta_k \leq O(\sqrt{n})\sigma_k$. Also, if  $S \cap X_k \neq \emptyset$, then Lemma~\ref{lem:crossover} implies that $f(S) \geq \sigma_k$. We can then use the fact that $\sigma_k$'s are geometric to show that the total value of the parts $Z_k$ that intersect $S$ is at most $O(\sqrt{n})f(S)$.

Now, consider the heavy types in $X_k$, i.e.~those types $j$ with $w_j\delta_k > \sigma_k/\sqrt{n}$. We further partition these types according to their weights in powers of 2. Let $R_{k,i} = \{j \in X_k \setminus Z_k : w_j \in [2^i, 2^{i+1})\}$. For each weight class $i$, we greedily partition $R_{k,i}$ into as many parts of size $\lceil \sqrt{n} \rceil$ as we can. This produces a collection $F_{k,i}$ of parts of size $\lceil \sqrt{n} \rceil$ and at most one leftover part $G_{k,i}$ of size less than $\sqrt{n}$. We say that a part is \emph{nice} if it belongs to $F_{k,i}$ and the part $G_{k,i}$ a \emph{leftover} part. 

Observe that there are at most $\lceil \sqrt{n} \rceil$ nice parts, each of size at most $\lceil \sqrt{n} \rceil$ and contains types of roughly the same weight. Thus, we can use a similar argument as in the unweighted case to show that the total value of the nice parts that intersect $S$ is at most $O(\sqrt{n})f(S)$. For the leftover parts, we charge the parts $G_{k,i}$ that intersect $S$ with $k < \ell$ to $w(S)\delta_\ell$ and those with $k \geq \ell$ to $\sigma_\ell$.

\begin{algorithm}[H]
\caption{\label{alg:concave_partition} Weighted Symmetric Subadditive Partitioning Algorithm}

\begin{algorithmic}[1]
\FOR{$k=1$ to $p$}
\STATE Create a part $Z_k = \{j \in X_k : w_j \delta_k \leq \sigma_k/\sqrt{n}\}$ \label{line:defn_zk}
\STATE Let $R_{k,i} = \{j \in X_k \setminus Z_k : w_j \in [2^i, 2^{i+1})\}$
\FOR{each $i$}
\STATE Greedily partition $R_{k,i}$ into as many sets of size exactly $\lceil \sqrt{n} \rceil$ as possible
\STATE Let $F_{k,i}$ denote the sets of size of size  $\lceil \sqrt{n} \rceil$
\STATE Let $G_{k,i}$ denote the remaining set of size less than $\sqrt{n}$, if it exists
\STATE Create a part for each set in $F_{k,i}$ and a part for the set $G_{k,i}$
\ENDFOR
\ENDFOR
\end{algorithmic}
\end{algorithm}

We now give the detailed analysis below.

\begin{theoremapprep}\label{thm:weighted-concave}
    For any weighted symmetric subadditive service function $f$, there exists a disjoint service function $g$ that $O(\sqrt{n})$-approximates $f$.
    It can be found in time polynomial w.r.t.\ the weights defining $f$.
\end{theoremapprep}
\begin{proof}
    
Let $S$ be an arbitrary set and suppose $f(S) = \min_{1 \leq i \leq p} \sigma_i + w(S)\delta_i = \sigma_\ell + w(S)\delta_\ell$. We now decompose $w(S)$ using the partition produced by our algorithm. In particular, we have

\[f(S) = \sigma_\ell + \left(\sum_k w(Z_k \cap S) + \sum_{k, i} \sum_{T \in F'_{k,i}} w(T \cap S) + \sum_{k, i} w(G_{k,i} \cap S)\right) \cdot \delta_\ell.\] 

Define $F'_{k, i}$ as the subset of parts in $F_{k, i}$ that intersects with $S$. 
We now show that the algorithm pays at most $O(\sqrt{n})f(S)$. In other words, we will prove that the total value of the parts that intersect $S$ are upper bounded as follows:
\[\sum_{k : Z_k \cap S \neq \emptyset} f(Z_k) + \sum_{k, i} \sum_{T \in F'_{k,i}} f(T) + \sum_{k, i : G_{k,i} \cap S \neq \emptyset} f(G_{k,i})\leq O(\sqrt{n})f(S).\]

\newcommand{\kmax}{k_{\operatorname{max}}}

 We begin by bounding $\sum_{k : Z_k \cap S \neq \emptyset} f(Z_k)$.  Let $\kmax$ be the largest $k$ such that $Z_k \cap S \neq \emptyset$. (If none  exists, then we do not need to bound this term.) We have that $w(S)\cdot \delta_{\kmax-1} \geq w(Z_k \cap S)\cdot \delta_{\kmax-1} \geq \sigma_{\kmax}$. Thus, Lemma~\ref{lem:crossover} implies that $f(S) \geq \sigma_{\kmax}$. On the other hand, 
\begin{align*}
    \sum_{k : Z_k \cap S \neq \emptyset} f(Z_k)
    \leq \sum_{k : Z_k \cap S \neq \emptyset} O(\sqrt{n}) \sigma_k 
    \leq O(\sqrt{n}) \sigma_{\kmax} \leq O(\sqrt{n})f(S).
\end{align*}
where the first inequality follows directly from the definition of $Z_k$ in line~\ref{line:defn_zk} of Algorithm~\ref{alg:concave_partition} and since there are at most $n$ elements in $Z_k$, the second inequality is since the $\sigma_k$'s are geometrically increasing.

Next, we bound $\sum_{k, i} \sum_{T \in F'_{k,i}} f(T)$. Since every set $T \in F'_{k,i}$ has size  $\lceil \sqrt{n} \rceil$, we have $\sum_{k,i} |F'_{k,i}| \leq \sqrt{n}$. Moreover, every $j \in T$ has $w_j \in [2^i, 2^{i+1})$, so $w(T) \leq O(\sqrt{n}) w(T \cap S)$. Thus, we have
\begin{align*}
    \sum_{k, i} \sum_{T \in F'_{k,i}} f(T)
    &\leq \sum_{k, i} \sum_{T \in F'_{k,i}} \sigma_\ell + w(T)\delta_\ell \\
    &\leq \sum_{k,i} |F'_{k,i}| \sigma_\ell + O(\sqrt{n})\sum_{k, i} \sum_{T \in F'_{k,i}}w(T \cap S) \delta_\ell\\
    &\leq O(\sqrt{n}) \left(\sigma_\ell + \sum_{k, i} \sum_{T \in F'_{k,i}} w(T \cap S) \delta_\ell\right)
    \leq O(\sqrt{n}) f(S).
\end{align*}
where the last inequality follows from the fact that all $T \in F^\prime_{k,i}$ are disjoint so we have that $w(S) \geq \sum_{k, i} \sum_{T \in F'_{k,i}} w(T \cap S)$.

We now turn to bounding $\sum_{k, i : G_{k,i} \cap S \neq \emptyset} f(G_{k,i})$. Consider a set $G_{k,i}$ that intersects $S$ for $\ell \leq k \leq p$. Since $G_{k,i}$ is a subset of $X_k \setminus Z_k$ and is at most of size $\sqrt{n}$, we have that $$f(G_{k,i}) \leq \sigma_k + w(G_{k,i})\delta_k \leq O(\sqrt{n}) w(G_{k,i} \cap S) \delta_k.$$ Since $\delta_k \leq \delta_\ell$, we get that 
\[\sum_{k\geq \ell} \sum_{i : G_{k,i} \cap S \neq \emptyset} f(G_{k,i}) \leq 
\sum_{k\geq \ell} \sum_{i : G_{k,i} \cap S \neq \emptyset} O(\sqrt{n}) w(G_{k,i} \cap S)\delta_\ell \leq O(\sqrt{n}) f(S).\]

Finally, when $\ell = 1$, the argument is complete. Let us now consider the case when $\ell > 1$. Consider a set $G_{k,i}$ that intersects $S$ for $1 \leq k < \ell$. We have that $f(G_{k,i}) \leq \sigma_k + w(G_{k,i})\delta_k \leq O(\sqrt{n}) 2^{i+1} \delta_k$. Moreover, since every $j \in X_k$ has $w_j \delta_k < \sigma_{k+1}$, we have that  
\[\sum_{k < \ell} \sum_{i : G_{k,i} \cap S \neq \emptyset} f(G_{k,i})
\leq O(\sqrt{n}) \sum_{k < \ell} \sigma_{k+1} \leq O(\sqrt{n}) \sigma_\ell \leq O(\sqrt{n})f(S).\]

Finally, it is not hard to see that, by design, Algorithm~\ref{alg:concave_partition} can be implemented in polynomial time in the logarithm of the total weight, $\log(w(U))$. This concludes the proof.
\end{proof}

Thus, by Lemma \ref{lem:reduction-to-disjoint}, we get Theorem~\ref{thm:non-clairvoyant-cjrp}.

\begin{toappendix}

\section{Tight Instances against Previous Algorithms}

\subsection{An $\Omega(\sqrt{n\log n})$ Tight Instance for the  Algorithm of~\cite{JiaLNRS05}}

\begin{proposition}
\label{thm:tight}
There exists a weighted set cover instance for which the Universal Set Cover algorithm of \cite{JiaLNRS05} has stretch $\Omega(\sqrt{n \log n})$.
\end{proposition}

\begin{proof}
    The algorithm of \cite{JiaLNRS05} works as follows: while the set $U$ of elements $e$ for  which $f(e)$ is undefined is non-empty, pick the set $S$ that minimizes $\frac{c(S)}{\sqrt{|S \cap U|}}$ and for all $e \in S \cap U$, define $f(e) = S$.

The high-level idea is that \cite{JiaLNRS05}'s analysis uses the Cauchy-Schwarz inequality and the tight instance is created by looking at when the Cauchy-Schwarz inequality is tight. 

Consider the following set system where we have sets $S, S_1, \ldots, S_k$ for some $k$ that we will choose later. The set $S$ contains $k$ elements and set $S_i$ contains $\left\lfloor\frac{k}{k - (i-1)}\right\rfloor$ elements. The sets also satisfy that $|S \cap S_i| = 1$ and $S_i \cap S_j = \emptyset$ for $1 \leq i < j \leq k$.  Moreover, the sets $S_i$ form a partition of all the $n$ elements. The costs of the sets are: $c(S) = 1$, $c(S_i) = \frac{\sqrt{|S_i|}}{\sqrt{k-(i-1)}}$. 

We now claim that in the $i$-th iteration, the algorithm chooses $S_i$. First observe that  for $1 \leq i < j \leq k$, we have
\[\frac{c(S_i)}{\sqrt{|S_i|}} < \frac{c(S_j)}{\sqrt{|S_j|}}.\]
Thus, it suffices to show that in each iteration $i$, the algorithm chooses $S_i$ over $S$. We do this by induction on $i$. When $i=1$, we have that 
\[\frac{c(S_1)}{\sqrt{|S_1|}} = \frac{1}{\sqrt{k}} = \frac{c(S)}{\sqrt{|S|}}.\] Now consider $i > 1$. By induction, we have that $|S \cap U| = k - (i-1)$ and $S_i \cap U = S_i$ (the latter is because the only set that intersects $S_i$ is $S$). Thus, we also have
\[\frac{c(S_i)}{\sqrt{|S_i|}} = \frac{1}{\sqrt{k - (i-1)}} = \frac{c(S)}{\sqrt{|S \cap U|}}.\] We conclude that in each iteration $i$, the algorithm chooses $S_i$.

Thus, the competitive ratio of the algorithm is at least 
$(\sum_{i=1}^k c(S_i))/c(S) = \sum_{i=1}^k c(S_i)$ since $c(S) = 1$. We have that 
\begin{equation}
\label{eq:k}
    \sum_{i=1}^k c(S_i) =\sum_{i=1}^k \frac{\sqrt{\left\lfloor\frac{k}{k - (i-1)}\right\rfloor}}{\sqrt{k-(i-1)}} = 
\Omega(\sqrt{k} \log k). 
\end{equation}

It now remains to maximize $k$. The constraint on $k$ is that $\sum_{i =1}^k |S_i| = n$ since $S_i$s are disjoint. Now, $\sum_{i=1}^k |S_i| = \sum_{i=1}^k \left\lfloor\frac{k}{k - (i-1)}\right\rfloor = \Theta(k \log k)$. Thus, setting $k = \Theta(n/\log n)$ satisfies the constraint on $k$. Plugging this into \eqref{eq:k} yields the claim.
\end{proof}

 \end{toappendix}
\begin{toappendix}
\subsection{An $\Omega(\sqrt{n\log n})$ Tight Instance for the Algorithm of~\cite{Touitou23}}

We complement the $O(\sqrt{n})$-stretch achieved by Algorithm~\ref{alg:mla_partition} and Algorithm~\ref{alg:concave_partition} with a JRP instance such that the algorithm of \cite{Touitou23} (Algorithm 2) must suffer a stretch of at least $\Omega(\sqrt{n\log n})$. Note that the instance we present in Figure~\ref{fig:mla-touitou} is both an MLA instance and a weighted concave one. This shows that for the specific case of MLA and weighted concave functions, not only is our algorithm optimal, but also that Touitou's algorithm cannot achieve the same guarantee. At a high level, whenever Touitou's algorithm decides to serve some requests, it issues up to two services (lines 9 and 12). One of them serves a subset of requests $R$ for which delay and service costs are the same. At the same time, a second service with a budget of up to $\sqrt{n \log n} \cdot c(R)$ can be issued to serve some pending requests in advance. The following example is one where the optimal algorithm rarely issues this second service.

\begin{figure}
    \centering
    \resizebox{0.2\linewidth}{!}{

\tikzset{every picture/.style={line width=0.75pt}} 

\begin{tikzpicture}[x=0.75pt,y=0.75pt,yscale=-1,xscale=1]

\draw   (312,35) .. controls (312,32.24) and (314.24,30) .. (317,30) .. controls (319.76,30) and (322,32.24) .. (322,35) .. controls (322,37.76) and (319.76,40) .. (317,40) .. controls (314.24,40) and (312,37.76) .. (312,35) -- cycle ;
\draw   (312,88) .. controls (312,85.24) and (314.24,83) .. (317,83) .. controls (319.76,83) and (322,85.24) .. (322,88) .. controls (322,90.76) and (319.76,93) .. (317,93) .. controls (314.24,93) and (312,90.76) .. (312,88) -- cycle ;
\draw   (363,123) .. controls (363,120.24) and (365.24,118) .. (368,118) .. controls (370.76,118) and (373,120.24) .. (373,123) .. controls (373,125.76) and (370.76,128) .. (368,128) .. controls (365.24,128) and (363,125.76) .. (363,123) -- cycle ;
\draw   (261,123) .. controls (261,120.24) and (263.24,118) .. (266,118) .. controls (268.76,118) and (271,120.24) .. (271,123) .. controls (271,125.76) and (268.76,128) .. (266,128) .. controls (263.24,128) and (261,125.76) .. (261,123) -- cycle ;
\draw    (317,93) -- (368,118) ;
\draw    (317,93) -- (266,118) ;
\draw    (317,40) -- (317,83) ;

\draw (287,111.4) node [anchor=north west][inner sep=0.75pt]    {$\dotsc \dotsc \dotsc $};
\draw (351,93.4) node [anchor=north west][inner sep=0.75pt]    {$w$};
\draw (273,92.4) node [anchor=north west][inner sep=0.75pt]    {$w$};
\draw (324,53.4) node [anchor=north west][inner sep=0.75pt]    {$1$};
\draw (259,135.4) node [anchor=north west][inner sep=0.75pt]    {$v_{2}$};
\draw (360,134.4) node [anchor=north west][inner sep=0.75pt]    {$v_{n}$};
\draw (310,5.4) node [anchor=north west][inner sep=0.75pt]    {$v_{0}$};
\draw (327,76.4) node [anchor=north west][inner sep=0.75pt]    {$v_{1}$};

\end{tikzpicture} }
    \caption{Tight instance for \cite{Touitou23}, where $w = \frac{\sqrt{n \log n} - 1}{n-1}$.}
    \label{fig:mla-touitou}
\end{figure}
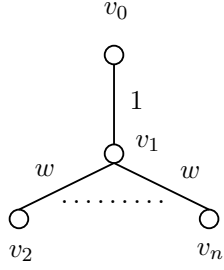

\begin{proposition}
\label{prop:tight-touitou}
There exists an instance for which the algorithm of \cite{Touitou23} has stretch $\Omega(\sqrt{n \log n})$. Moreover, this is an MLA and a weighted concave instance. 
\end{proposition}

\begin{proof}
Let us consider the JRP tree $T$ in Figure~\ref{fig:mla-touitou}, where $w = \frac{\sqrt{n \log n} - 1}{n-1}$ and the delay cost functions on the nodes read
\[  
    d_i(t) = \begin{cases}
        2t, \text{ if } i = 1\\
        \varepsilon t, \text{ if } i \geq 2
    \end{cases},
\]
for $\varepsilon \ll w$ to be set later. In particular, at each time step, there are $n$ requests arriving on tree $T$, one per node. 

Let us first observe that the optimum algorithm only serves the requests at $v_1$, paying a service cost of $1$ at each time step. Moreover, it serves requests arriving at any $v_i$ with $i \geq 2$ once $\varepsilon t = w$, i.e., every $w/\varepsilon$ time steps, and pays $(n-1)w + 1 = \sqrt{n \log n}$. Thus, letting $\tau$ be the length of the requests sequence, the overall optimal cost is $\OPT(\tau) = \tau + \frac{\varepsilon \tau}{w} \sqrt{n \log n} \leq 2\tau$, by setting $\varepsilon = w/n$.

Algorithm 2 in \cite{Touitou23} (whose cost is referred to as $\ALG$ from now on) serves a request arriving at $v_1$ (line 9) as soon as its accumulated delay equals its service cost (this is when the {\scshape UponCritical} event occurs). Once a request at $v_1$ arrives, the algorithm waits until the time elapsed $t$ is such that $2t = 1$ to serve it. That is, when the $j$-th request located at $v_1$ arrives, the algorithm serves it at time $t_j = j + \frac{1}{2}$. Right after, it issues a second service (line 12) to serve all other requests at $v_2, \ldots, v_n$. Overall, the algorithm pays $\ALG(\tau) = \tau \cdot (1 + (n-1)w) = \tau \sqrt{n \log n}$. 

Hence, $$\frac{\ALG(\tau)}{\OPT(\tau)} \geq \frac{\sqrt{n \log n}}{2},$$ for all $\tau \geq 1$. To conclude, the fact that the instance in Figure~\ref{fig:mla-touitou} is an MLA one comes directly from the fact that it is a depth 2 tree. Moreover, observe that no matter how we choose $S \subseteq V$, $f(S) = f(w(S))$, and thus the instance in Figure~\ref{fig:mla-touitou} is also a weighted concave instance.
\end{proof} \end{toappendix}

\section*{Acknowledgments}
Tomer Ezra is supported by Harvard University Center of Mathematical Sciences and Applications.
Stefano Leonardi and Matteo Russo are partially supported by the ERC Advanced Grant 788893 AMDROMA ``Algorithmic and Mechanism Design Research in Online Markets'' and MIUR PRIN project ALGADIMAR ``Algorithms, Games, and Digital Markets''.
Part of this work was done when Seeun William Umboh was visiting the Sapienza University of Rome, and at the School of Computer Science, University of Sydney.

\clearpage
\bibliographystyle{alpha}
\bibliography{sample}

\newcommand{\etalchar}[1]{$^{#1}$}
\begin{thebibliography}{KMWW20}

\bibitem[AAC{\etalchar{+}}17]{bipartite}
Itai Ashlagi, Yossi Azar, Moses Charikar, Ashish Chiplunkar, Ofir Geri, Haim
  Kaplan, Rahul Makhijani, Yuyi Wang, and Roger Wattenhofer.
\newblock Min-cost bipartite perfect matching with delays.
\newblock In {\em Approximation, Randomization, and Combinatorial Optimization.
  Algorithms and Techniques, {APPROX/RANDOM} 2017}, volume~81 of {\em LIPIcs},
  pages 1:1--1:20. Schloss Dagstuhl - Leibniz-Zentrum f{\"{u}}r Informatik,
  2017.

\bibitem[ACK17]{polylog}
Yossi Azar, Ashish Chiplunkar, and Haim Kaplan.
\newblock Polylogarithmic bounds on the competitiveness of min-cost perfect
  matching with delays.
\newblock In {\em Proceedings of the 28th Annual {ACM-SIAM} Symposium on
  Discrete Algorithms, {SODA} 2017}, pages 1051--1061. {SIAM}, 2017.

\bibitem[ACKT20]{AzarCKT20}
Yossi Azar, Ashish Chiplunkar, Shay Kutten, and Noam Touitou.
\newblock Set cover with delay - clairvoyance is not required.
\newblock In Fabrizio Grandoni, Grzegorz Herman, and Peter Sanders, editors,
  {\em 28th Annual European Symposium on Algorithms, {ESA} 2020, September 7-9,
  2020, Pisa, Italy (Virtual Conference)}, volume 173 of {\em LIPIcs}, pages
  8:1--8:21. Schloss Dagstuhl - Leibniz-Zentrum f{\"{u}}r Informatik, 2020.

\bibitem[AF20]{hemispheres}
Yossi Azar and Amit~Jacob Fanani.
\newblock Deterministic min-cost matching with delays.
\newblock {\em Theory Comput. Syst.}, 64(4):572--592, 2020.

\bibitem[AGGP21]{OSD}
Yossi Azar, Arun Ganesh, Rong Ge, and Debmalya Panigrahi.
\newblock Online service with delay.
\newblock {\em {ACM} Trans. Algorithms}, 17(3):23:1--23:31, 2021.

\bibitem[ARV21]{concave}
Yossi Azar, Runtian Ren, and Danny Vainstein.
\newblock The min-cost matching with concave delays problem.
\newblock In {\em Proceedings of the 2021 {ACM-SIAM} Symposium on Discrete
  Algorithms, {SODA} 2021}, pages 301--320. {SIAM}, 2021.

\bibitem[AT19]{AzarT19}
Yossi Azar and Noam Touitou.
\newblock General framework for metric optimization problems with delay or with
  deadlines.
\newblock In David Zuckerman, editor, {\em 60th {IEEE} Annual Symposium on
  Foundations of Computer Science, {FOCS} 2019, Baltimore, Maryland, USA,
  November 9-12, 2019}, pages 60--71. {IEEE} Computer Society, 2019.

\bibitem[AT20]{AzarT20}
Yossi Azar and Noam Touitou.
\newblock Beyond tree embeddings - a deterministic framework for network design
  with deadlines or delay.
\newblock In {\em 61st {IEEE} Annual Symposium on Foundations of Computer
  Science, {FOCS} 2020, Durham, NC, USA, November 16-19, 2020}, pages
  1368--1379. {IEEE}, 2020.

\bibitem[BBB{\etalchar{+}}20]{BienkowskiBBCDF20}
Marcin Bienkowski, Martin B{\"{o}}hm, Jaroslaw Byrka, Marek Chrobak, Christoph
  D{\"{u}}rr, Luk\'a\v{s} Folwarczn\'y, Lukasz Jez, Jir{\'{\i}} Sgall,
  Kim~Thang Nguyen, and Pavel Vesel{\'{y}}.
\newblock Online algorithms for multilevel aggregation.
\newblock {\em Oper. Res.}, 68(1):214--232, 2020.

\bibitem[BBB{\etalchar{+}}21]{BienkowskiBBCDF21}
Marcin Bienkowski, Martin B{\"{o}}hm, Jaroslaw Byrka, Marek Chrobak, Christoph
  D{\"{u}}rr, Luk\'a\v{s} Folwarczn\'y, Lukasz Jez, Jir{\'{\i}} Sgall,
  Kim~Thang Nguyen, and Pavel Vesel{\'{y}}.
\newblock New results on multi-level aggregation.
\newblock {\em Theor. Comput. Sci.}, 861:133--143, 2021.

\bibitem[BBBM]{FL-linear}
Marcin Bienkowski, Martin B{\"{o}}hm, Jaroslaw Byrka, and Jan Marcinkowski.
\newblock Online facility location with linear delay.
\newblock In {\em APPROX 2022}.

\bibitem[BBC{\etalchar{+}}14]{BienkowskiBCJS13}
Marcin Bienkowski, Jaroslaw Byrka, Marek Chrobak, Lukasz Jez, Dorian Nogneng,
  and Jir{\'{\i}} Sgall.
\newblock Better approximation bounds for the joint replenishment problem.
\newblock In Chandra Chekuri, editor, {\em Proceedings of the Twenty-Fifth
  Annual {ACM-SIAM} Symposium on Discrete Algorithms, {SODA} 2014, Portland,
  Oregon, USA, January 5-7, 2014}, pages 42--54. {SIAM}, 2014.

\bibitem[BDF{\etalchar{+}}12]{BadanidiyuruDFKNR12}
Ashwinkumar Badanidiyuru, Shahar Dobzinski, Hu~Fu, Robert Kleinberg, Noam
  Nisan, and Tim Roughgarden.
\newblock Sketching valuation functions.
\newblock In {\em {SODA}}, pages 1025--1035. {SIAM}, 2012.

\bibitem[BFNT17]{BuchbinderFNT17}
Niv Buchbinder, Moran Feldman, Joseph~(Seffi) Naor, and Ohad Talmon.
\newblock \emph{O}(depth)-competitive algorithm for online multi-level
  aggregation.
\newblock In Philip~N. Klein, editor, {\em Proceedings of the Twenty-Eighth
  Annual {ACM-SIAM} Symposium on Discrete Algorithms, {SODA} 2017, Barcelona,
  Spain, Hotel Porta Fira, January 16-19}, pages 1235--1244. {SIAM}, 2017.

\bibitem[BJN07]{BuchbinderJN07}
Niv Buchbinder, Kamal Jain, and Joseph Naor.
\newblock Online primal-dual algorithms for maximizing ad-auctions revenue.
\newblock In Lars Arge, Michael Hoffmann, and Emo Welzl, editors, {\em
  Algorithms - {ESA} 2007, 15th Annual European Symposium, Eilat, Israel,
  October 8-10, 2007, Proceedings}, volume 4698 of {\em Lecture Notes in
  Computer Science}, pages 253--264. Springer, 2007.

\bibitem[BKL{\etalchar{+}}08]{BuchbinderKLMS08}
Niv Buchbinder, Tracy Kimbrel, Retsef Levi, Konstantin Makarychev, and Maxim
  Sviridenko.
\newblock Online make-to-order joint replenishment model: primal dual
  competitive algorithms.
\newblock In Shang{-}Hua Teng, editor, {\em Proceedings of the Nineteenth
  Annual {ACM-SIAM} Symposium on Discrete Algorithms, {SODA} 2008, San
  Francisco, California, USA, January 20-22, 2008}, pages 952--961. {SIAM},
  2008.

\bibitem[BKLS18]{pd}
Marcin Bienkowski, Artur Kraska, Hsiang{-}Hsuan Liu, and Pawel Schmidt.
\newblock A primal-dual online deterministic algorithm for matching with
  delays.
\newblock In {\em Approximation and Online Algorithms - Proceedings of the 16th
  International Workshop, {WAOA} 2018, Helsinki, Finland, August 23-24, 2018,
  Revised Selected Papers}, volume 11312 of {\em Lecture Notes in Computer
  Science}, pages 51--68. Springer, 2018.

\bibitem[BKS17]{spheres}
Marcin Bienkowski, Artur Kraska, and Pawel Schmidt.
\newblock A match in time saves nine: Deterministic online matching with
  delays.
\newblock In {\em Approximation and Online Algorithms - Proceedings of the 15th
  International Workshop, {WAOA} 2017}, volume 10787 of {\em Lecture Notes in
  Computer Science}, pages 132--146. Springer, 2017.

\bibitem[BKS18]{OSD-line}
Marcin Bienkowski, Artur Kraska, and Pawel Schmidt.
\newblock Online service with delay on a line.
\newblock In Zvi Lotker and Boaz Patt{-}Shamir, editors, {\em Structural
  Information and Communication Complexity - 25th International Colloquium,
  {SIROCCO} 2018, Ma'ale HaHamisha, Israel, June 18-21, 2018, Revised Selected
  Papers}, volume 11085 of {\em Lecture Notes in Computer Science}, pages
  237--248. Springer, 2018.

\bibitem[BKV12]{BritoKV12}
Carlos~Fisch Brito, Elias Koutsoupias, and Shailesh Vaya.
\newblock Competitive analysis of organization networks or multicast
  acknowledgment: How much to wait?
\newblock {\em Algorithmica}, 64(4):584--605, 2012.

\bibitem[BR11]{BhawalkarR11}
Kshipra Bhawalkar and Tim Roughgarden.
\newblock Welfare guarantees for combinatorial auctions with item bidding.
\newblock In {\em {SODA}}, pages 700--709. {SIAM}, 2011.

\bibitem[CKU22]{CJRP}
Ryder Chen, Jahanvi Khatkar, and Seeun~William Umboh.
\newblock Online weighted cardinality joint replenishment problem with delay.
\newblock In Mikolaj Bojanczyk, Emanuela Merelli, and David~P. Woodruff,
  editors, {\em 49th International Colloquium on Automata, Languages, and
  Programming, {ICALP} 2022, July 4-8, 2022, Paris, France}, volume 229 of {\em
  LIPIcs}, pages 40:1--40:18. Schloss Dagstuhl - Leibniz-Zentrum f{\"{u}}r
  Informatik, 2022.

\bibitem[CPSV18]{CarrascoPSV18}
Rodrigo~A. Carrasco, Kirk Pruhs, Cliff Stein, and Jos{\'{e}} Verschae.
\newblock The online set aggregation problem.
\newblock In Michael~A. Bender, Martin Farach{-}Colton, and Miguel~A. Mosteiro,
  editors, {\em {LATIN} 2018: Theoretical Informatics - 13th Latin American
  Symposium, Buenos Aires, Argentina, April 16-19, 2018, Proceedings}, volume
  10807 of {\em Lecture Notes in Computer Science}, pages 245--259. Springer,
  2018.

\bibitem[DFF21]{DobzinskiFF21}
Shahar Dobzinski, Uriel Feige, and Michal Feldman.
\newblock Are gross substitutes a substitute for submodular valuations?
\newblock In {\em {EC}}, pages 390--408. {ACM}, 2021.

\bibitem[DGS01]{DoolyGS01}
Daniel~R. Dooly, Sally~A. Goldman, and Stephen~D. Scott.
\newblock On-line analysis of the {TCP} acknowledgment delay problem.
\newblock {\em J. {ACM}}, 48(2):243--273, 2001.

\bibitem[Dob07]{Dobzinski07}
Shahar Dobzinski.
\newblock Two randomized mechanisms for combinatorial auctions.
\newblock In {\em {APPROX-RANDOM}}, volume 4627 of {\em Lecture Notes in
  Computer Science}, pages 89--103. Springer, 2007.

\bibitem[DU23]{DeryckereU23}
Lindsey Deryckere and Seeun~William Umboh.
\newblock Online matching with set and concave delays.
\newblock In Nicole Megow and Adam~D. Smith, editors, {\em Approximation,
  Randomization, and Combinatorial Optimization. Algorithms and Techniques,
  {APPROX/RANDOM} 2023, September 11-13, 2023, Atlanta, Georgia, {USA}}, volume
  275 of {\em LIPIcs}, pages 17:1--17:17. Schloss Dagstuhl - Leibniz-Zentrum
  f{\"{u}}r Informatik, 2023.

\bibitem[EFRS20]{EzraFRS20}
Tomer Ezra, Michal Feldman, Tim Roughgarden, and Warut Suksompong.
\newblock Pricing multi-unit markets.
\newblock {\em {ACM} Trans. Economics and Comput.}, 7(4):20:1--20:29, 2020.

\bibitem[EKW16]{emekoriginal}
Yuval Emek, Shay Kutten, and Roger Wattenhofer.
\newblock Online matching: haste makes waste!
\newblock In {\em Proceedings of the 48th Annual {ACM} {SIGACT} Symposium on
  Theory of Computing, {STOC} 2016}, pages 333--344. {ACM}, 2016.

\bibitem[ESW19]{2sources}
Yuval Emek, Yaacov Shapiro, and Yuyi Wang.
\newblock Minimum cost perfect matching with delays for two sources.
\newblock {\em Theor. Comput. Sci.}, 754:122--129, 2019.

\bibitem[GHIM09]{GoemansHIM09}
Michel~X. Goemans, Nicholas J.~A. Harvey, Satoru Iwata, and Vahab~S. Mirrokni.
\newblock Approximating submodular functions everywhere.
\newblock In {\em {SODA}}, pages 535--544. {SIAM}, 2009.

\bibitem[GKP20]{cachingDelay}
Anupam Gupta, Amit Kumar, and Debmalya Panigrahi.
\newblock Caching with time windows.
\newblock In Konstantin Makarychev, Yury Makarychev, Madhur Tulsiani, Gautam
  Kamath, and Julia Chuzhoy, editors, {\em Proceedings of the 52nd Annual {ACM}
  {SIGACT} Symposium on Theory of Computing, {STOC} 2020, Chicago, IL, USA,
  June 22-26, 2020}, pages 1125--1138. {ACM}, 2020.

\bibitem[GKP21]{hitting-kserver}
Anupam Gupta, Amit Kumar, and Debmalya Panigrahi.
\newblock A hitting set relaxation for {\textdollar}k{\textdollar}-server and
  an extension to time-windows.
\newblock In {\em 62nd {IEEE} Annual Symposium on Foundations of Computer
  Science, {FOCS} 2021, Denver, CO, USA, February 7-10, 2022}, pages 504--515.
  {IEEE}, 2021.

\bibitem[GMM09]{GuhaMM09}
Sudipto Guha, Adam Meyerson, and Kamesh Munagala.
\newblock A constant factor approximation for the single sink edge installation
  problem.
\newblock {\em {SIAM} J. Comput.}, 38(6):2426--2442, 2009.

\bibitem[JLN{\etalchar{+}}05]{JiaLNRS05}
Lujun Jia, Guolong Lin, Guevara Noubir, Rajmohan Rajaraman, and Ravi Sundaram.
\newblock Universal approximations for tsp, steiner tree, and set cover.
\newblock In Harold~N. Gabow and Ronald Fagin, editors, {\em Proceedings of the
  37th Annual {ACM} Symposium on Theory of Computing, Baltimore, MD, USA, May
  22-24, 2005}, pages 386--395. {ACM}, 2005.

\bibitem[KC82]{KelsoC82}
Alexander~S. Kelso and Vincent~P. Crawford.
\newblock Job matching, coalition formation, and gross substitutes.
\newblock {\em Econometrica}, 50(6):1483--1504, 1982.

\bibitem[KKR01]{dynamictcp}
Anna~R. Karlin, Claire Kenyon, and Dana Randall.
\newblock Dynamic {TCP} acknowledgement and other stories about e/(e-1).
\newblock In Jeffrey~Scott Vitter, Paul~G. Spirakis, and Mihalis Yannakakis,
  editors, {\em Proceedings on 33rd Annual {ACM} Symposium on Theory of
  Computing, July 6-8, 2001, Heraklion, Crete, Greece}, pages 502--509. {ACM},
  2001.

\bibitem[KMWW20]{nonclairvoyantKServer}
Predrag Krnetic, Darya Melnyk, Yuyi Wang, and Roger Wattenhofer.
\newblock The k-server problem with delays on the uniform metric space.
\newblock In Yixin Cao, Siu{-}Wing Cheng, and Minming Li, editors, {\em 31st
  International Symposium on Algorithms and Computation, {ISAAC} 2020, December
  14-18, 2020, Hong Kong, China (Virtual Conference)}, volume 181 of {\em
  LIPIcs}, pages 61:1--61:13. Schloss Dagstuhl - Leibniz-Zentrum f{\"{u}}r
  Informatik, 2020.

\bibitem[LLN06]{LehmannLN06}
Benny Lehmann, Daniel Lehmann, and Noam Nisan.
\newblock Combinatorial auctions with decreasing marginal utilities.
\newblock {\em Games Econ. Behav.}, 55(2):270--296, 2006.

\bibitem[LPWW18]{convex}
Xingwu Liu, Zhida Pan, Yuyi Wang, and Roger Wattenhofer.
\newblock Impatient online matching.
\newblock In {\em Proceedings of the 29th International Symposium on Algorithms
  and Computation, {ISAAC} 2018}, volume 123 of {\em LIPIcs}, pages
  62:1--62:12. Schloss Dagstuhl - Leibniz-Zentrum f{\"{u}}r Informatik, 2018.

\bibitem[LUX23]{LeUX23}
Ngoc~Mai Le, Seeun~William Umboh, and Ningyuan Xie.
\newblock The power of clairvoyance for multi-level aggregation and set cover
  with delay.
\newblock In Nikhil Bansal and Viswanath Nagarajan, editors, {\em Proceedings
  of the 2023 {ACM-SIAM} Symposium on Discrete Algorithms, {SODA} 2023,
  Florence, Italy, January 22-25, 2023}, pages 1594--1610. {SIAM}, 2023.

\bibitem[McM21]{McMahan-MLA}
Jeremy McMahan.
\newblock A d-competitive algorithm for the multilevel aggregation problem with
  deadlines.
\newblock {\em CoRR}, abs/2108.04422, 2021.

\bibitem[Tou21]{Touitou21}
Noam Touitou.
\newblock Nearly-tight lower bounds for set cover and network design with
  deadlines/delay.
\newblock In Hee{-}Kap Ahn and Kunihiko Sadakane, editors, {\em 32nd
  International Symposium on Algorithms and Computation, {ISAAC} 2021, December
  6-8, 2021, Fukuoka, Japan}, volume 212 of {\em LIPIcs}, pages 53:1--53:16.
  Schloss Dagstuhl - Leibniz-Zentrum f{\"{u}}r Informatik, 2021.

\bibitem[Tou23a]{Touitou23}
Noam Touitou.
\newblock Frameworks for nonclairvoyant network design with deadlines or delay.
\newblock In {\em 50th International Colloquium on Automata, Languages, and
  Programming, {ICALP} 2023}, page to appear, 2023.

\bibitem[Tou23b]{Touitou23-OSD}
Noam Touitou.
\newblock Improved and deterministic online service with deadlines or delay.
\newblock In Barna Saha and Rocco~A. Servedio, editors, {\em Proceedings of the
  55th Annual {ACM} Symposium on Theory of Computing, {STOC} 2023, Orlando, FL,
  USA, June 20-23, 2023}, pages 761--774. {ACM}, 2023.

\end{thebibliography}

\end{document}